\documentclass[]{article}

\usepackage{amsmath}
\usepackage{amssymb}
\usepackage{bm}
\usepackage{mathrsfs}
\usepackage{amsmath,amssymb,amsthm}
\usepackage{fullpage}
\usepackage{algpseudocode}
\usepackage{algorithm}

\usepackage[pdftex,                %
    bookmarks         = true,
    bookmarksnumbered = true,
    pdfpagemode       = None,
    pdfstartview      = FitH,
    pdfpagelayout     = SinglePage,
    colorlinks        = true,
    linkcolor= blue, 
    anchorcolor= blue, 
    citecolor         =blue,
    urlcolor          = magenta,
    pdfborder         = {0 0 0}
    ]{hyperref}

\hypersetup{
linkcolor=blue,
citecolor=red, 
urlcolor=blue, 
}
\hypersetup{colorlinks=true}
\makeatletter

\def\mdeg{\ensuremath{\mathrm{mdeg}}}

\def\hgt{\ensuremath{\mathrm{ht}}}
\def\coeff{\ensuremath{\mathrm{coeff}}}

\def\A {\ensuremath{\mathbb{A}}}
\def\C {\ensuremath{\mathbb{C}}}
\def\R {\ensuremath{\mathbb{R}}}
\def\N {\ensuremath{\mathbb{N}}}
\def\Z {\ensuremath{\mathbb{Z}}}
\def\K {\ensuremath{\mathbb{K}}}
\def\FF {\ensuremath{\mathbb{F}}}
\def\P {\ensuremath{\mathbb{P}}}
\def\Q {\ensuremath{\mathbb{Q}}}

\def\Qbar {\ensuremath{\overline{\mathbb{Q}}}}
\def\Kbar {\ensuremath{\overline{\mathbb{K}}}}

\def\dunder{\underline{d}}
\def\mydegh{d}
\def\mydeg{\mathfrak{d}}
\def\myheight{s}
\def\homotop{\mathbf{\mathrm{\bf homot}}}
\def\myxi{{\xi}}

\def\GS {\ensuremath{\mathsf{G}}}

\def\mA{\ensuremath{\mathbf{A}}}

\def\GL{\ensuremath{{\rm GL}}}

\def\reg{\ensuremath{{\rm reg}}}

\def\rank{\ensuremath{{\rm rank}}}
\def\jac{\ensuremath{{\mathrm{jac}}}}

\def\scrQ{\ensuremath{\mathscr{Q}}}

\def\scrC{\ensuremath{\mathscr{C}}}
\def\scrH{\ensuremath{\mathscr{H}}}

\def\g {\ensuremath{\mathbf{g}}}
\def\h {\ensuremath{\mathbf{h}}}

\def\s {\ensuremath{\mathbf{s}}}
\def\L {\ensuremath{\mathbb{L}}}
\def\bL {\ensuremath{\mathbf{L}}}
\def\F {\ensuremath{\mathbb{F}}}

\def\X {\ensuremath{\mathbf{X}}}
\def\x {\ensuremath{\mathbf{x}}}

\def\u {\ensuremath{\mathbf{u}}}

\DeclareBoldMathCommand{\bn}{n}
\DeclareBoldMathCommand{\d}{d}
\DeclareBoldMathCommand{\bc}{c}
\DeclareBoldMathCommand{\bh}{h}
\DeclareBoldMathCommand{\bu}{u}
\DeclareBoldMathCommand{\bx}{x}
\DeclareBoldMathCommand{\bs}{s}
\DeclareBoldMathCommand{\btheta}{\vartheta}
\DeclareBoldMathCommand{\f}{f}
\DeclareBoldMathCommand{\bell}{\ell}
\DeclareBoldMathCommand{\bbeta}{\eta}
\DeclareBoldMathCommand{\bchi}{\chi}

\newtheorem{theorem}{Theorem}
\newtheorem{corollary}[theorem]{Corollary}
\newtheorem{proposition}[theorem]{Proposition}
\newtheorem{lemma}[theorem]{Lemma}
\newtheorem{remark}[theorem]{Remark}

\usepackage{color}

\begin{document}

\title{Bit complexity for multi-homogeneous polynomial system solving
  \\Application to polynomial minimization}

\author{M. {Safey El Din}$^{1}$, \'E. Schost$^{2}$}

\footnotetext[1]{Sorbonne Universit\'es, UPMC Univ. Paris 06, CNRS, INRIA Paris Center, LIP6, PolSys Team, France}
\footnotetext[2]{David Cheriton School of Computer Science, University of Waterloo, ON, Canada}

\maketitle

\begin{abstract}
  Multi-homogeneous polynomial systems arise in many applications.  We
  provide bit complexity estimates for solving them which, up to a few
  extra other factors, are quadratic in the number of solutions and
  linear in the height of the input system, under some genericity
  assumptions. The assumptions essentially imply that the Jacobian
  matrix of the system under study has maximal rank at the solution
  set and that this solution set is finite. The algorithm is
  probabilistic and a probability analysis is provided.

  Next, we apply these results to the problem of optimizing a linear
  map on the real trace of an algebraic set. Under some genericity
  assumptions, we provide bit complexity estimates for solving this
  polynomial minimization problem.
\end{abstract}



\section{Introduction}


\subsection{Motivation and problem statement}

In this paper, we are interested in exact algorithms solving systems
of polynomial equations with a multi-homogeneous structure (the
polynomials we consider are actually affine, but can be seen as the
dehomogenization of multi-homogeneous ones); we focus in particular on
the bit complexity aspects of this question. The main application we
have in mind is the solution of some constrained optimization
problems. This is used in many algorithms for studying real solutions
to polynomial systems (see e.g. \cite{BaGiHeMb97, BaGiHeMb01, SaSc03,
  BGHS14, SaSc13} and references therein). We will also pay 
particular attention to the situation when the constraints are given
as quadratic equations.

We work with polynomials in $m$ groups of variables.  Let thus
$\bn=(n_1,\dots,n_m)$ be positive integers, and consider variables
$\X=(\X_1,\dots,\X_m)$, with $\X_1=(X_{1,1},\dots,X_{1,n_1})$, \dots,
$\X_m=(X_{m,1},\dots,X_{m,n_m})$. We write $N=n_1 + \cdots + n_m$ for
the total number of variables.

Let $\K$ be a field and $\f=(f_1,\dots,f_M)$ in $\K[\X_1,\dots,\X_m]$,
for some $M \le N$ (we will sometimes write $\f_M$ instead of $\f$, in
order to highlight the length of the sequence). We associate to $\f$
the algebraic set $Z(\f)$, defined as the set of all $\bx $ in
$\Kbar{}^N$ such that $\f(\bx)=0$ and such that the Jacobian matrix of
$\f$ has rank $M$ at $\bx$.  By the Jacobian
criterion~\cite[Chapter~16]{Eisenbud95}, $Z(\f)$ is either empty, or
equidimensional of dimension $N-M$, and it is defined over $\K$.

Suppose that $M=N$.  It is known that using the multi-degree structure
of $\f$, that is, the partial degrees of these equations in
$\X_1,\dots,\X_m$, together with a multi-homogeneous B\'ezout bound,
we can obtain finer estimates on the cardinality of $Z(\f)$ than
through the direct application of B\'ezout's theorem in many cases.

In this paper, we focus on the case $\K=\Q$, and show how the same
phenomenon holds in terms of bit complexity. Indeed, our goal is to obtain an
algorithm for solving such systems whose bit complexity is, up to some
extra factors, quadratic in the multi-homogeneous bound and linear in
the {\em heights} of the polynomials in the input system (which is a
measure of their bit size).


In the following paragraphs, we recall the notion of height and the
data structure we use to represent $Z(\f)$. We will also use these
notions to describe related works on solving multi-homogeneous
systems.

\smallskip Let us first however describe how such results can be
applied to the problem of minimizing the map
$\pi_1: (x_1, \ldots, x_n)\mapsto x_1$ subject to the constraints
$h_1=\cdots=h_p=0$, with
$\h=(h_1, \ldots, h_p)\subset \Z[X_1, \ldots, X_n]$. Assuming that
$\h$ is a reduced regular sequence, that the minimizer exists and that
the set of minimizers is finite, it is well-known that this problem
can be tackled by solving the so-called Lagrange system
\[
h_1=\cdots=h_p=0, \ [L_1, \ldots, L_p]
\begin{bmatrix}
  \frac{\partial h_1}{\partial X_2} & \cdots &   \frac{\partial h_1}{\partial X_n} \\
  \vdots & & \vdots \\
  \frac{\partial h_p}{\partial X_2} & \cdots &   \frac{\partial h_p}{\partial X_n} \\
\end{bmatrix} = [0 \ \cdots \ 0 ], u_1L_1+\cdots+u_pL_p=1,
\]
where $\bL=(L_1, \ldots, L_p)$ are new variables (called Lagrange
multipliers) and $(u_1, \ldots, u_p)$ are
randomly chosen integers.  Hence, using the notation introduced above, we have
for this system $m=2$, $\bn=(n, p)$, $\X=(\X_1, \X_2)$ with
$\X_1=(X_1, \ldots, X_n)$ and $\X_2=\bL$.


\subsection{Bit size and data structures} \label{intro:notation}

\subsubsection{Multi-degree, height and bit size}

Let $\K$ be a field as above. To a polynomial $f$ in
$\K[\X_1,\dots,\X_m]$ we associate its {\em multi-degree}
$\mdeg(f)=(d_1,\dots,d_m) \in \N^m$, with $d_i=\deg(f,\X_i)$ for all
$i$. When comparing multi-degrees, we use the (partial) componentwise
order, so that saying that $f$ has multi-degree at most
$\dunder=(d_1,\dots,d_m)$ means that $\deg(f,\X_i) \le d_i$ holds for
all $i$. Similarly, to a sequence of polynomials $\f_M$, we associate its
multi-degree $\mdeg(\f_M)=(\mdeg(f_1),\dots,\mdeg(f_M))$.  Saying that
$\f_M$ has multi-degree at most $\d=(\dunder_1,\dots,\dunder_M)$, with
now all $\dunder_i=(d_{i,1},\dots,d_{i,m})$ in $\N^m$, means that
$\deg(f_i,\X_j) \le d_{i,j}$ holds for all $i,j$.

\smallskip

Consider a polynomial $f$ with coefficients in $\Q$. To measure its
bit size, we will use its {\em height}, defined as follows. First, for
$a=u/v$ in $\Q-\{0\}$, define the height of $a$, $\hgt(a)$, as
$\max(\log(|u|), \log(v))$, with $u \in \Z$ and $v \in \N$ coprime
. For a non-zero univariate or multivariate
polynomial $f$ with rational coefficients, we let $v \in \N$ be the
minimal common denominator of all its non-zero coefficients; then
$\hgt(f)$ is defined as the maximum of the logarithms of $v$ and of the
absolute values of the coefficients of $v f$ (which are integers). 

When $f$ has integer coefficients, this is simply the maximum of the
logarithms of the absolute values of these coefficients.  More
generally, for $f$ with rational coefficients, knowing the degree and
height of a polynomial with rational coefficients gives us an upper
bound on the size of its binary representation. As in the case of
degrees, for polynomials $\f_M=(f_1,\dots,f_M)$, we write that
$\hgt(\f_M)=(\hgt(f_1),\dots,\hgt(f_M))$, and we say that $\hgt(\f_M)
\le \bs$, with $\bs=(s_1,\dots,s_M)$, if $\hgt(f_i) \le s_i$ holds for
all $i$.

Given $\bbeta=(\eta_1, \ldots, \eta_M)$ in $\R^M$ and
$\d=(\dunder_1, \ldots, \dunder_M)$ with
$\dunder_i=(d_{1,1}, \ldots, d_{i,m})\in \N^m$, we denote by
$\scrC_\bn(\d)$ the sum of the coefficients of the polynomial 
\[
\prod_{i=1}^M(d_{i,1}\vartheta_1+\cdots+d_{i,m}\vartheta_m)\mod \langle \vartheta_1^{n_1+1},\dots,\vartheta_m^{n_m+1} \rangle
\]
and by $\scrH_{\bn}(\bbeta,\d)$ the sum of the coefficients of the
polynomial
\[
\prod_{i=1}^M(\eta_i\zeta+d_{i,1}\vartheta_1+\cdots+d_{i,m}\vartheta_m)\mod \langle \zeta^2,\vartheta_1^{n_1+1},\dots,\vartheta_m^{n_m+1}\rangle.
\]

\subsubsection{Zero-dimensional parametrizations}

Consider a zero-dimensional algebraic set $V \subset \Kbar{}^N$,
defined over $\K$. A {\em zero-dimensional parametrization}
$\scrQ=((q,v_1,\dots,v_N),\lambda)$ of $V$ consists in polynomials
$(q,v_1,\dots,v_N)$, such that $q\in \K[T]$ is monic and squarefree,
all $v_i$'s are in $\K[T]$ and satisfy $\deg(v_i) < \deg(q)$, and in a
$\K$-linear form $\lambda$ in $N$ variables, such that

\smallskip

\begin{itemize}
\item $\lambda(v_1,\dots,v_N)=T q' \bmod q$, where
  $q'=\frac{\partial q}{\partial T}$;

\smallskip

\item we have the equality
  $V=\left \{\left(
      \frac{v_1(\tau)}{q'(\tau)},\dots,\frac{v_N(\tau)}{q'(\tau)}\right
    ) \ \mid \ q(\tau)=0 \right \};$
\end{itemize}

\smallskip

\noindent the constraint on $\lambda$ then says that the roots of $q$ are
precisely the values taken by $\lambda$ on $V$. This definition
implies that the linear form $\lambda$ takes pairwise distinct values
on the points of $V$; we call such linear forms {\em separating}
and
we say that $\scrQ$ is {\em associated to} $\lambda$.

\smallskip This data structure has a long history, going back to work
of Kronecker and Macaulay \cite{Kronecker82,Macaulay16}, and has been
used in a host of algorithms in effective
algebra~\cite{GiMo89,GiHeMoPa95,ABRW,GiHeMoPa97,GiHeMoMoPa98,Rouillier99,GiLeSa01,Lecerf2000}.

The reason for using a rational parametrization with $q'$ as a
denominator is well-known~\cite{ABRW, Rouillier99, GiLeSa01}: when
$\K=\Q$, and for systems without necessarily any kind of
multi-homogeneous structure, it leads to a precise theoretical control
on the size of the coefficients, which is verified in practice
extremely accurately.  A main purpose of this article is to show how
such results, which are known for general systems, can be extended and
refined to take into account multi-homogeneous situations.

\subsection{Main results}

\subsubsection{Algorithm for solving multi-homogeneous polynomial systems}

The main result of the paper is a probabilistic algorithm for solving
multi-homogeneous systems.  Following references such
as~\cite{GiHeMoPa95,GiHeMoPa97,GiHeMoMoPa98,GiLeSa01,Lecerf2000}, we
will represent the input polynomials $\f$ of our algorithm by means of
a {\em straight-line program}, that is, a sequence of elementary
operations $+,-,\times$ that evaluates the polynomials $\f$ from the
input variables $\X_1,\dots,\X_m$; the {\em length} or {\em size} $L$ of such an objet
is simply the number of operations it performs.

The approach developed here is not new: we start by computing a
zero-dimensional parametrization of $Z(\f \bmod p)$, for a well-chosen
prime $p$, and lift it modulo powers of $p$ to a zero-dimensional
paramet\-rization of $Z(\f)$. The novelty of the theorem below lies in
the use of multi-homogeneous height bounds proved hereafter to control
the cost of the process.

The algorithm is randomized, and part of the randomness amounts to
choosing the prime $p$. Constructing primes is a difficult question in
itself, and not the topic of this paper; hence, we will assume that we
are given an oracle $\mathscr{O}$, which takes as input an integer $B$
and returns a prime number in $\{B+1,\dots,2B\}$, uniformly
distributed within the set of primes in this interval (for a
randomized solution to this question, we refer the reader
to~\cite[Section~18.4]{GaGe03}). In all the paper, we use the
soft-$O$ notation $O\tilde{~}$, in order to indicate that we omit
polylogarithmic terms.


\begin{theorem}\label{thm:homoZ}
  Suppose that $\f=(f_1,\dots,f_N)$ satisfies $\mdeg(\f) \le
  \d=(\dunder_1,\dots,\dunder_m)$ and $\hgt(\f)\le \bs=(s_1,\dots,s_N)$, and that $\f$
  is given by means of a straight-line program $\Gamma$ of size $L$,
  that uses integer constants of height at most $b$.

  There exists an algorithm {\sf NonsingularSolutionsOverZ} that takes
  $\Gamma$, $\d$ and $\bs$ as input, and that produces one of the
  following outputs:

\smallskip

  \begin{itemize}
  \item either a zero-dimensional parametrization of $Z(\f)$,

\smallskip

  \item or a zero-dimensional parametrization of degree less than that of $Z(\f)$,

\smallskip
  \item or {\sf fail}.
  \end{itemize}

\smallskip

\noindent  The first outcome occurs with probability at least $21/32$. In any
  case, the algorithm uses
  $$ O\tilde{~}\left (L b + \scrC_\bn(\d)\scrH_{\bn}(\bbeta,\d) \left (L+ N\mydeg +N^2 \right) N(\log(\myheight) + N )
  \right)$$
  boolean operations, with
  $$
  \mydeg = \max_{1 \le i \le N} d_{i,1} + \cdots + d_{i,m},\quad
  \myheight =\max_{1 \le i \le N}(s_i) \quad\text{and}\quad
  \bbeta=\left (s_i + \sum_{j=1}^m \log(n_j+1) d_{i,j}\right)_{1 \le i \le N}.
  $$
The
  algorithm calls the oracle $\mathscr{O}$ with an input parameter
  $B=\myheight \mydeg^{O(N)}$ and the polynomials in the output have
  degree at most $\scrC_\bn(\d)$ and height
  $O\tilde{~}(\scrH_{\bn}(\bbeta,\d) + N \scrC_\bn(\d))$.
\end{theorem}

A more detailed discussion on the probabilistic aspects and the choice
of the prime $p$ is given in Subsection~\ref{sec:mainalgo}. Here, we investigate an immediate consequence of Theorem~\ref{thm:homoZ},
when for all $1\leq i \leq N$, we have
$\mdeg(f_i)\leq \dunder = (d_1, \ldots, d_m)$ and $\hgt(f_i)\leq \myheight$, that is,
$\d=(\dunder, \ldots, \dunder)$ and $\bs=(s,\dots,s)$.
Technical but immediate computations show that in this case,
$$\scrC_\bn(\d)= d_1^{n_1}\cdots d_m^{n_m}{{N}\choose{n_1\cdots
    n_m}}$$
where ${{N}\choose{n_1\cdots n_m}}$ is the multinomial coefficient
$\frac{N!}{n_1!\cdots n_m!}$ (recall that $N=n_1+\cdots+n_m$) and
\[
\scrH_{\bn}(\bbeta,\d) \leq m(\myheight+\mydeg+1) d_1^{n_1}\cdots d_m^{n_m}
{{N}\choose{n_1\cdots n_m}}
\]
where $\mydeg=d_1+\cdots+d_m$.  From this, we obtain the following corollary.

\begin{corollary}
  Suppose that $\f=(f_1,\dots,f_N)$ satisfies
  $\mdeg(f_i) \le \dunder=(d_1, \ldots, d_m)$ and
  $\hgt(f_i)\le \myheight$ for all $i$, and that $\f$ is given by
  means of a straight-line program $\Gamma$ of size $L$, that uses
  integer constants of height at most $b$.

  There exists an algorithm {\sf NonsingularSolutionsOverZ} that takes
  $\Gamma$, $\dunder$ and $\myheight$ as input, and that produces one of the
  following outputs:

\smallskip

  \begin{itemize}
  \item either a zero-dimensional parametrization of $Z(\f)$,

\smallskip

  \item or a zero-dimensional parametrization of degree less than that of $Z(\f)$,

\smallskip

  \item or {\sf fail}.
  \end{itemize}

\smallskip

\noindent  The first outcome occurs with probability at least $21/32$. In any
  case, the algorithm uses
  $$ O\tilde{~}\left (L b + 
    \left (d_1^{n_1}\cdots d_m^{n_m}{{N}\choose{n_1\cdots
    n_m}}\right )^2m (\myheight+\mydeg)
\left (L+ N\mydeg +N^2 \right) N(\log(\myheight) + N ) \right)$$
boolean operations, with
$\mydeg =  d_{1} + \cdots + d_{m}$. The
algorithm calls the oracle $\mathscr{O}$ with an input parameter
$B=\myheight \mydeg^{O(N)}$ and the polynomials in the output have
degree at most $d_1^{n_1}\cdots d_m^{n_m}{{N}\choose{n_1\cdots n_m}}$
and height $O\tilde{~}(
d_1^{n_1}\cdots d_m^{n_m}{{N}\choose{n_1\cdots n_m}}(m(\myheight + \mydeg)+N)
)$.
\end{corollary}



\subsubsection{Minimization problems} 

We describe now the main results on {\em generic instances} of the
problem of minimizing the map $\pi_1: (x_1, \ldots, x_n)\mapsto x_1$
subject to the constraints $h_1=\cdots=h_p=0$, with $\h=(h_1, \ldots,
h_p)\subset \Z[X_1, \ldots, X_n]$, using our algorithm for
multi-homogeneous systems.

This is done by considering the Lagrange system in $N=n+p$ variables
$$
h_1=\cdots=h_p=0, \qquad L_1\frac{\partial h_1}{\partial
  X_j}+\cdots+L_p\frac{\partial h_p}{\partial X_j}=0 \text{ for }
2\leq j \leq n, \qquad u_1L_1+\cdots+u_pL_p=1
$$
where $\bL=L_1, \ldots, L_p$ are new variables and
$(u_1, \ldots, u_p)$ are chosen at random. As we will
see, in generic situations, the projection on the
$(X_1, \ldots, X_n)$-space of the complex solution set of this system
is finite, and coincides with the set of critical points of $\pi_1$ on
$V=V(h_1,\dots,h_p)$.  

Let $\mydegh$ be the maximum of the degrees of the polynomials in $\h$.
The Lagrange system above possesses a bi-homogeneous structure, with $p$
equations of total degree at most $\mydegh$, resp.\ $0$ in variables
$\X$, resp.\ $\bL$ (we will then speak of bidegree $(d,0)$),
$n-1$ equations of bidegree at most $(\mydegh-1,1)$ and one equation of
bidegree $(0,1)$.

We prove in Section~\ref{sec:appli} that we can solve the bi-homogeneous
system above in randomized time
$$ O\tilde{~}\left ( p (E+n) \myheight' + 
{{n-1} \choose {p-1}}{n \choose p}
  (\myheight+\mydegh)\mydegh^{2p}(\mydegh-1)^{2(n-p)}(pE+n\mydegh+n^2)
\right ),$$
where $\myheight$ is the height of the input polynomials, $E$ is the
length of the straight-line program that computes them, and
$\myheight'$ is the height of the integers that appear in this
straight-line program (in most cases, one expects
$\myheight' \le \myheight$, in which case the first term
disappears). The degree $\scrC$ of the output is at most
${{n-1} \choose {p-1}}\mydegh^p(\mydegh-1)^{n-p}$, and its height
$\scrH$ is
$ O\tilde{~}\left (n{n \choose
    p}(\myheight+\mydegh)\mydegh^p(\mydegh-1)^{n-p}\right )$.

One can always construct a naive straight-line program for the input
polynomials, simply by computing all monomials they involve and
summing them. In this case, one can take
$E \in O(p{{n+\mydegh}\choose{n}})$ and $\myheight'=\myheight$, which
leads to a boolean runtime of the form
$$ O\tilde{~}\left ( {{n-1} \choose {p-1}}{n \choose p}{{n+\mydegh}\choose{n}}
  (\myheight+\mydegh)\mydegh^{2p}(\mydegh-1)^{2(n-p)} \right ).$$
Taking $p=1$ as in~\cite{MeSa16}, this is
$ O\tilde{~}\left ( (\myheight+\mydegh)\mydegh^{n+2}(\mydegh-1)^{2(n-1)}
\right ).$
In this case, for large $\mydegh$, our result is hardly an improvement
over the cost $O\tilde{~}( \mydegh^{3n}  \myheight)$ obtained in that
reference.

The gain is much more significant in the case $\mydegh=2$. In this case, we
can take $E\in O(pn^2)$ and $\myheight'=\myheight$. As a result, we
obtain a running time of
$ O\tilde{~} (n^{5} {{n-1} \choose {p-1}}{n \choose p}\myheight 2^{2p} )$
for the quadratic case, for an output of degree $\scrC$ at most
${{n-1} \choose {p-1}} 2^p$, and of height $\scrH$ in
$O\tilde{~} (n {n \choose p}\myheight 2^p )$: when the codimension $p$ is
fixed, all these quantities are {\em polynomial} in $n$, with 
the runtime being $ O\tilde{~} (n^{2p+4}\myheight)$.

We end this section with an easy consequence of the above result,
concerning the determination of an isolating interval for
$\min_{\bx\in V\cap \R^n} \pi_1(\bx)$.  The output of our algorithm
describes a finite set in the $\X,\bL$-space whose projection on the
$\X$-space is the set of critical points of $\pi_1$ on $V$. From the
zero-dimensional parametrization of this set, using root isolation
algorithms as in~\cite[Section~3]{MeSa16}, we can then compute boxes
of side length $2^{-\sigma}$ around all roots of the system using
$O\tilde{~}(n \scrC^2 \scrH + n \scrC \sigma)$ bit operations, with
$\scrC$ and $\scrH$ the bounds on the output degree and height
mentioned above. For instance, in the quadratic case, this is
$O\tilde{~}( n^2{{n-1} \choose {p-1}}^2{n \choose p} \myheight 2^{3p}
+ n{{n-1} \choose {p-1}} 2^p \sigma )$ bit operations. For fixed $p$,
the cost of the root isolation step is $O\tilde{~}( n^{2p+1} \myheight
+ n^p \sigma )$, so the whole process is polynomial in~$n$.

As an illustration/application, one may mention the {\em
  Celis-Dennis-Tapia (CDT) problem}~\cite{Celi84}, to minimize a
non-convex quadratic function over the intersection of two ellipsoids,
which can be turned into an instance of the problem above by
introducing a new dummy variable.  Such problems arise naturally in
iterative non-linear optimization procedures where in one iteration
step, the objective function and the constraints are approximated by
quadratic models. Taking $p=3$ as in the CDT problem, the
overall cost for computing a zero-dimensional parametrization of the 
minimizers and computing isolating boxes is $O\tilde{~}(n^{10} \myheight + n^3 \sigma)$.

\subsection{Related work} 

\subsubsection{Multi-homogeneous polynomial systems}

As already said, the techniques used in the algorithm are not new: we
first solve the system modulo a prime, using a symbolic homotopy
algorithm that adapts to the multi-homogeneous case an algorithm given
by Jeronimo {\it et al.}~\cite{JeMaSoWa09} for the sparse case; then,
we use lifting techniques from~\cite{GiLeSa01,Schost03}, as well as
techniques coming from \cite[Section~4]{RRS}, to recover the output
over $\Z$. Taking into account our upper bound on the height of the
output, this results in the first bound (that we are aware of) on the
boolean cost of solving polynomial systems that involves their
multi-homogeneous structure in such a manner. Our results on the
heights of zero-dimensional parametrizations computed by our algorithm
rely on objects introduced by, and results due to D'Andrea, Krick and
Sombra~\cite{DaKrSo13}.

Although we do not have boolean complexity bounds to compare with,
several results are known in an arithmetic complexity model (where
we count base field operations at unit cost).
In the bi-homogeneous case, the algorithm in~\cite{HeJeSaSo02} has an
 arithmetic cost at least $\scrC_\bn(\d)^5\scrC_{\bn'}(\d')^6$
with $\bn'=(1, n_1, \ldots n_m)$ and $\d'=(\dunder'_1, \ldots,
\dunder'_m)$, where for all $i$ we set $\dunder'_i=(1, d_{i,1}, \ldots, d_{i, m})$.
Closer to us are two algorithms from~\cite{GiLeSa01}
and~\cite{JeMaSoWa09}. The geometric resolution algorithm
of~\cite{GiLeSa01} solves our questions in time quadratic in a
particular geometric degree associated to the input system; however,
in general, this degree cannot be controlled in terms of the
quantities $\scrC_\bn(\d)$ and $\scrC_{\bn'}(\d')$ used in our
analysis (see for example those systems appearing in \cite{HNS16}); in
addition, we are not aware of a probability analysis for it.

Another line of work exploits properties of resultant formulae to
solve multi-homogene\-ous systems; we refer in particular to
\cite{JeSa07, EM12, GKP13} among many others and we also mention
\cite{EMT16} focusing on the particular case of bilinear systems. In
this setting, solving multi-homogeneous polynomial systems mostly reduce
to compute determinants of structured submatrices of the Macaulay
matrix. The bit complexity results obtained this way are cubic in
$\scrC_\bn(\d)$; exploiting the structure of Macaulay submatrices, we
do not know whether a result essentially linear in
$\scrC_\bn(\d)\scrH_\bn(\d)$, such as the one  in Theorem~\ref{thm:homoZ},
could be obtained in this formalism.

\subsubsection{Minimization problems}

We comment now on related work on minimization problems.  If we let
$\mydegh$ be the maximum of the degrees of the input polynomials, it is
known that the critical point method runs in time
$\mydegh^{O(n)}$~\cite[Section~14.2]{BaPoRo06} in an {\em algebraic
  complexity model}, counting arithmetic operations in the base field
$\Q$ at unit cost.

More precisely, using Gr\"obner bases techniques,
papers~\cite{DBLP:conf/issac/FaugerePS12} and~\cite{Sp14} establish
that if the polynomials $h_1, \ldots, h_p$ are generic enough, this
computation can done using
\[ O\left ({{n+D_\reg}\choose{n}}^\omega+\left
    (\mydegh^p(\mydegh-1)^{n-p}{{n-1}\choose{p-1}}\right )^\omega\right
) \]
operations in $\Q$, with $D_\reg=\mydegh(p-1)+(\mydegh-2)n+2$, and where
$\omega$ is such that computing the row echelon form of a matrix of
size $k\times k$ is done in time $O(k^\omega)$. In the quadratic case,
with $\mydegh=2$,
this becomes
\[ O\left ({{n+2p}\choose{2p}}^\omega+\left
    (2^{p}{{n-1}\choose{p-1}}\right )^\omega\right )\subset
O((n+2p)^{2p\omega}) \] operations in $\Q$.  The best known value for
$\omega$ is $\omega < 2.38$~\cite{LeGall14}; in the often discussed
case where $p$ is constant, the cost is then $O(n^{4.76p})$. For the
CDT problem, we have $p=3$, so that generic instances of it can be
solved using $O(n^{14.28})$ arithmetic operations.


The quadratic case has actually been known to be solvable in
$n^{O(p^2)}$ bit operations since Barvinok's paper~\cite{Barv93}; this
was later improved to $n^{O(p)}$ by Grigoriev and Vorobjov in
\cite{GP}. The algorithms are deterministic, and make no assumption on
the input system, but the constant in the big-$O$ exponent is not
specified. In~\cite{JePe14}, Jeronimo and Perrucci give a randomized
algorithm to compute the minimum of a function on a basic
semi-algebraic set. In our setting, with $\myheight=2$ and $p$ fixed,
the running time is $O\tilde{~}(n^{2p+5} + n^{3p})$ arithmetic
operations.

Fewer references discuss bit complexity. When $p=1$, \cite[Prop. 3.8
  and Lemma 4.1]{MeSa16} give boolean complexity estimates of the form
$O\tilde{~}(\myheight \mydegh^{3n})$ for critical point computation on
a hypersurface, under some genericity assumptions on the input; here,
$\myheight$ is an upper bound on the height of the input
polynomials. Height bounds on the minimum polynomial defining
$\min_{\bx \in V\cap \R^n}\pi_1(\bx)$ are given in~\cite{JPT}; they
turn out to be of the same order as the ones we derive, but no
algorithm with bit complexity depending on these bounds is
given. 


\subsection{Plan of the paper} We start by recalling basic notions
and fixing notation in Section~\ref{sec:prelim}. In particular, this
section states height bounds for the output of our algorithms; the
proof of these bounds is postponed to the end of the paper in
Section~\ref{proof:prop:hgt}. Section~\ref{sec:degreebounds} gives a
symbolic homotopy deformation algorithm dedicated to multi-homogeneous
cases; in the main algorithm, we apply this result over a prime
field. Section~\ref{sec:bounds} discusses computations over the
rationals, with a cost analysis in the boolean model. We finally apply
this to our minimization problem in Section~\ref{sec:appli}.


\paragraph{Acknowledgments.}The first author is member of and
supported by Institut Universitaire de France. The second author is
supported by an NSERC Discovery Grant.

\section{Notation and preliminaries}\label{sec:prelim}

\subsection{Basic notions}

In the whole paper, we use freely basic notions such as dimension,
degree, reducibility and irreducibility, smoothness\dots, of algebraic
sets. We recall these basic notions below and we refer the reader to
references such as~\cite{ZaSa58,Mumford76,Shafarevich77,Eisenbud95}
for more details.

For a field $\K$ and an algebraic closure $\Kbar$ of $\K$, a
$\K$-algebraic set $V\subset \Kbar{}^N$ is the set of common solutions
in $\Kbar{}^N$ to $N$-variate polynomial equations with coefficients in
$\K$.  Usually, the base field $\K$ will be clear from the context; in
this case we simply say {\em algebraic sets} for $\K$-algebraic sets.

For a sequence of polynomials $\f_M=(f_1, \ldots, f_M)$ in the ring
of $N$-variate polynomials with coefficients in $\K$,
$V(\f)\subset \Kbar{}^N$ denotes the algebraic set defined by
$f_1=\cdots=f_M=0$. The ideal generated by $\f$ is denoted by
$\langle \f\rangle$. The ideal associated to $V(\f)$ is the set of
polynomials that vanish at all points of $V(\f)$.

For an algebraic set $V=V(\f)$, the dimension $\dim(V)$ of $V$ is
the Krull dimension of the coordinate ring of $V$; zero-dimensional
algebraic sets are non-empty finite algebraic sets. By convention,
the empty algebraic set has dimension $-1$.

When $V$ is an irreducible algebraic set, the degree of $V$ is the
number of points lying in the intersection of $V$ with $\dim(V)$
generic hyperplanes. The degree of an arbitrary algebraic set is the
sum of the degrees of its irreducible components. When the algebraic
set under consideration has dimension zero, its degree is its
cardinality.

An algebraic set $V=V(\f)\subset \Kbar{}^N$ is said to be
equidimensional when all its irreducible components have the same
dimension. In this case, assuming that $\f$ generates a radical
ideal, the smooth points of $V$ are those points at which the rank of
the Jacobian matrix of $\f$ is the codimension of $V$, {\it i.e.},
$N-\dim(V)$. Those points which are not smooth are called singular.

\subsection{Chow ring and arithmetic Chow ring}\label{ssec:heightbds}

We recall hereafter definitions for Chow rings and arithmetic Chow
rings; the latter ones are an arithmetic analogue to Chow rings due to
D'Andrea, Krick and Sombra~\cite{DaKrSo13}, on which most of our bit
size estimates will rely.

For a field $\K$, an algebraic closure $\Kbar$ of $\K$, 
and an $m$-uple $\bn=(n_1,\dots,n_m)$,
we denote
by $\P^{\bn}(\Kbar)$ the multi-projective space
$\P^{n_1}(\Kbar)\times \cdots \times \P^{n_m}(\Kbar)$. Consider the
ring of truncated power series
$$A^*(\P^\bn(\Kbar))=\Z[\vartheta_1, \dots,\vartheta_m]/\langle
\vartheta_1^{n_1+1},\dots,\vartheta_m^{n_m+1} \rangle;
$$ it is the {\em Chow ring} of the multi-projective space
$\P^{\bn}(\Kbar)$. For $\K=\Q$, we also
define
$$A^*(\P^\bn(\Qbar),\Z)=\R[\zeta, \vartheta_1, \dots,\vartheta_m]/\langle
\zeta^2, \vartheta_1^{n_1+1},\dots,\vartheta_m^{n_m+1} \rangle;
$$
this is called the {\em arithmetic Chow
  ring} of $\Q$. Since the field we use will be clear from the context, we
will use the simpler notations $A^*(\P^\bn)$ and $A^*(\P^\bn, \Z)$ for
Chow rings and arithmetic Chow rings.

Now, given a multi-degree
$\dunder=(d_1, \ldots, d_m)$ and  a non-negative real number $\eta$, we set 
\[
\chi(\dunder)=d_1\vartheta_1+\cdots+d_m\vartheta_m\in
A^*(\P^\bn)\]
and
\[
\chi'(\eta, \dunder)=\eta\zeta+d_1\vartheta_1+\cdots+d_m\vartheta_m\in
A^*(\P^\bn,\Z).\]
Given vectors $\d=(\dunder_1, \ldots, \dunder_M)$ and
$\bbeta=(\eta_1, \ldots, \eta_M)$, with all  $\dunder_i$ in $\N^m$
and all  $\eta_i$ in $\R_{\geq 0}$, we set 
\[
\bchi(\d)=\chi(\dunder_1)\cdots \chi(\dunder_M)
\in A^*(\P^\bn).
\]
and
\[
\bchi'(\bbeta, \d)=\chi(\eta_1, \dunder_1)\cdots \chi(\eta_M, \dunder_M)
\in A^*(\P^\bn,\Z).
\]
For $\bc=(c_1, \ldots, c_m)\in \N^m$, we denote in the sequel
$\vartheta_1^{c_1}\cdots\vartheta_m^{c_m}$ by $\vartheta^\bc$.  Note
that all monomials appearing in $\bchi(\d)$ and
$\bchi'(\bbeta,\d)$ have total degree $M$; then, we define the
quantities
\[
\scrC_\bn(\d) = \sum_{\bc \in \N^m} \coeff(\bchi(\d), \vartheta^\bc)
\]
and 
\[
\scrH_{\bn}(\bbeta,\d) = \sum\limits_{\bc \in \N^m,\ |\bc| =M-1}
\coeff(\bchi'(\bbeta,\d), \zeta\ \vartheta^\bc) + \sum_{\bc \in \N^m,\ |\bc| =M}
\coeff(\bchi'(\bbeta,\d), \vartheta^\bc).
\]
Note that they coincide with the quantities defined in
Subsection~\ref{intro:notation}. Observe also that all coefficients of
 $\bchi'(\bbeta,\d)$ not taken into account in the
above sums are necessarily zero.

The quantities $\scrC_\bn(\d)$ and $\scrH_\bn(\bbeta, \d)$ play a
crucial role for bounding the degree and the height of the output of
the algorithms described in the sequel. As an illustration, the
following degree inequality is proved in~\cite[Proposition I.1,
  electronic appendix]{SaSc13}. In what follows, we let 
$\X_1,\dots,\X_m$ be blocks of variables of respective lengths
$n_1,\dots,n_m$, as defined in the introduction.
\begin{proposition}\label{prop:degH}
  Let $\f=(f_1,\dots,f_M)$ be polynomials in $\K[\X_1,\dots,\X_m]$,
  with $\mdeg(\f) \le\d$. Then, the $(N-M)$-equidimensional component
  of $V(\f)$ has degree at most $\scrC_\bn(\d)$.
\end{proposition}
\noindent In particular, if $M=N$, $Z(\f)$ has degree (that is,
cardinality) at most $\scrC_\bn(\d)$, and thus all polynomials
appearing in a zero-dimensional parametrization of it have degree at
most $\scrC_\bn(\d)$. This latter claim is not new; see for
instance~\cite{MoSo87}.

All these definitions being written, we can state the new result of
this paragraph. Its proof is given in Section~\ref{proof:prop:hgt}.

\begin{proposition}\label{prop:hgt}
  Let $\f=(f_1,\dots,f_N)$ be polynomials in $\Z[\X_1,\dots,\X_m]$,
  with $\mdeg(\f) \le \d=(\dunder_1,\dots,\dunder_N)$ and
  $\dunder_i=(d_{i,1},\dots,d_{i,m})$ for all $i$, and $\hgt(\f) \le
  \bs=(s_1,\dots,s_N)$; let also $\lambda$ be a separating linear form
  for $Z(\f)$ with integer coefficients of height at most $b$. Then
  all polynomials in the zero-dimensional parametrization of $Z(\f)$
  associated to $\lambda$ have height at most $\scrH_\bn(\bbeta,\d) +
  (b + 4\log(N + 2))\scrC_\bn(\d),$ with
$$\bbeta=\left (s_i + \sum_{j=1}^m \log(n_j+1) d_{i,j}\right)_{1 \le i \le N}.$$
\end{proposition}


\section{The multi-homogeneous homotopy}\label{sec:degreebounds}

In this section, we work over a perfect field $\K$, using $N$
variables $\X=\X_1,\dots,\X_m$ partitioned into $m$ blocks of
respective lengths $(n_1,\dots,n_m)$, as explained in the
introduction. Our goal here is to give a symbolic homotopy algorithm
to compute $Z(\f)$, where $\f=(f_1,\dots,f_N)$ has coefficients in
$\K$, for use in the next section. These results are for a substantial
part not new. The algorithm can in particular be seen as a
modification of that in~\cite{JeMaSoWa09}; we do however have to give
a rather detailed presentation, for reasons explained in
Subsection~\ref{ssec:homo1}.
 



\subsection{Main statement}\label{ssec:homo1}

In order to compute a zero-dimensional parametrization of the
algebraic set $Z(\f)$, we use a symbolic adaptation of
multi-homogeneous homotopy continuation algorithms. In the context of
numerical continuation techniques, this approach is detailed
in~\cite{soWa05} and references therein; in a symbolic context, the
algorithm underlying the following proposition is inspired by e.g. the
algorithm in~\cite{HeJeSaSo02}, that applies in the bi-homogeneous
case.

We need here to introduce the following notation. Given a vector
$\d=(\dunder_1,\dots,\dunder_M)$, with
$\dunder_i=(d_{i,1},\dots,d_{i,m})$ for all $i$, we define the tuple
$\d'$ as $\d'=(\dunder'_1,\dots,\dunder'_M)$, with
$\dunder'_i = (1,d_{i,1},\dots,d_{i,m}) \in \N^{m+1}$ for all $i$,
together with $\bn'=(1,n_1,\dots,n_m)$. If we see $\d$ as being a
vector of multi-degrees, this corresponds to adding one new variable
(written $t$ below) and considering polynomials of degree $1$ in $t$
and multi-degree $\d$ in $\X_1,\dots,\X_m$. This allows us to
introduce the integer $\scrC_{\bn'}(\d')$, which we define as we did
for $\scrC_\bn(\d)$ above. Our convention was to use variables
$\vartheta_1,\dots,\vartheta_m$ for $\scrC_\bn(\d)$; to define
$\scrC_{\bn'}(\d')$, we introduce a new variable $\vartheta_0$ and let
$\scrC_{\bn'}(\d')$ be the sum of the coefficients of the polynomial
\[
\prod_{i=1}^M(\vartheta_0+d_{i,1}\vartheta_1+\cdots+d_{i,m}\vartheta_m)
\mod \langle \vartheta_0^2, \vartheta_1^{n_1+1}, \ldots,
\vartheta_m^{n_m+1}\rangle.
\]

\begin{proposition}\label{prop:homo1}
  Suppose that $\f=(f_1,\dots,f_N)$ has multi-degree at most
  $\d=(\dunder_1,\dots,\dunder_N)$, with all $\dunder_i$ in $\N^m$,
  and that $\f$ is given by a straight-line program $\Gamma$ of size
  $L$; suppose further that $\K$ has characteristic either zero or at
  least $e$, where
  $$e=\max \left(\max_{1 \le j \le m} d_{1,j} + \cdots + d_{N,j}, 8(N-1)\scrC_\bn(\d)^2\right).$$
  There exists an algorithm {\sf NonsingularSolutions} that takes
  $\Gamma$ and $\d$ as input and that outputs one of the following:

\smallskip

  \begin{itemize}
  \item either a zero-dimensional parametrization of $Z(\f)$,

\smallskip

  \item or a zero-dimensional parametrization of degree less than that of $Z(\f)$,

\smallskip
  \item or {\sf fail}.
  \end{itemize}

\smallskip

\noindent  The first outcome occurs with probability at least $7/8$.  In any
  case, the algorithm uses
  $$ O\tilde{~}\left (\scrC_\bn(\d)\scrC_{\bn'}(\d') \left (L+\sum_{1
    \le i \le N, 1 \le j \le m}d_{i,j}+N^2 \right)N \right)$$
  operations in $\K$, where we write $\d'=(d'_1,\dots,d'_N)$, with
  $d'_i = (1,d_{i,1},\dots,d_{i,m})$ for all $i$ and
  $\bn'=(1,n_1,\dots,n_m)$.
\end{proposition}
A discussion on probabilistic aspects and cases where the algorithm
fails is given in Remark~\ref{rem:failure}.

The algorithm of~\cite{JeMaSoWa09} deals with symbolic homotopies for
sparse systems, with a running time that would be comparable to ours
in the case of multi-homogeneous systems. However, that algorithm
requires a base field of characteristic zero (whereas we will need it
over a finite field), and the system $\f$ must be zero-dimensional
(which is not the case for us); in addition, the last step of that
algorithm, specialization at $t=1$ (Section~6.2 in~\cite{JeMaSoWa09})
appears to overlook issues that we discuss below, inspired
by~\cite[Section~4]{RRS}. For these reasons, lacking another
reference, we decided to include a self-contained proof dedicated to
our multi-homogeneous situation.

Without loss of generality, in what follows, we suppose that all
polynomials $f_i$ are non-constant.


\subsection{The start system}

The following construction is from~\cite{HeJeSaSo02, JeSa07} (however,
the cost estimates below are new). For any integers $i,j$, with $j$ in
$\{1,\dots,m\}$, let us define the affine polynomial
$$\kappa_i(\X_j) = X_{j,1} + i X_{j,2} + \cdots + i^{n_j-1} X_{j,n_j} + i^{n_j}.$$
Next, considering non-negative integers $\dunder=(d_1,\dots,d_m)$
and $\underline{e}=(e_1,\dots,e_m)$, we 
define the polynomial
$$g_{\dunder,\underline{e}} = \prod_{j=1}^m \prod_{k=0}^{d_j-1} \kappa_{k+e_j}(\X_j).$$
The following result is straightforward, once one notices that for any
$i$, $\kappa_i(\X_j)$ has multi-degree $(0,\dots,0,1,0,\dots,0)$, with $1$
at the $j$-th entry.
\begin{lemma}\label{lemma:2}
  The polynomial $g_{\dunder,\underline{e}}$ has multi-degree $\dunder$.
\end{lemma}
Finally, given multi-degrees $\d=(\dunder_1,\dots,\dunder_N)$, with
each $\dunder_i$ in $\N^m$, we define the system $\g=(g_1,\dots,g_N)$ by
$$g_i = g_{\dunder_i, \dunder_1 + \cdots + \dunder_{i-1}}
= \prod_{j=1}^m \prod_{k=0}^{d_{i,j}-1} \kappa_{k+d_{1,j}+\cdots+d_{i-1,j}}(\X_j), \qquad 1 \le i \le N.$$
\begin{lemma}\label{lemma:g}
  Suppose that $\K$ has characteristic zero, or at least $\max_{1 \le
    j \le m} d_{1,j} + \cdots + d_{N,j}$, and that for all $i$, $\dunder_i$
  is different from $(0,\dots,0)$. Then the following holds:

\smallskip

  \begin{itemize}
  \item for $i$ in $\{1,\dots,N\}$, $g_i$ has multi-degree $\dunder_i$;

\smallskip

  \item one can compute $\g$ by means of a straight-line program 
    of length $O\tilde{~}(\sum_{i,j} d_{i,j})$;

\smallskip

  \item $\g$ has $\scrC_\bn(\d)$ roots, and one can compute all of
    them using $O\tilde{~}(\scrC_\bn(\d) N)$ operations in $\K$;

\smallskip

  \item the Jacobian matrix of $\g$ is invertible at all these roots.
  \end{itemize}
\end{lemma}
\begin{proof}
  The first claim follows directly from Lemma~\ref{lemma:2}.
  In order to build a straight-line program for the polynomials $\g$, recall that $g_i(\X)$ takes the form
  $$g_i(\X) = \prod_{j=1}^m \prod_{k=0}^{d_{i,j}-1}
  \kappa_{k+d_{1,j}+\cdots+d_{i-1,j}}(\X_j).$$ We actually start by
  fixing $j$ in $\{1,\dots,m\}$. For such a fixed $j$, we have to
  evaluate all linear forms $\kappa_{k+d_{1,j}+\cdots+d_{i-1,j}}(\X_j)$,
  for $i=1,\dots,N$ and $k=0,\dots,d_{i,j}-1$. Due to the shape of
  these linear forms, each such evaluation amounts to computing the
  value of the polynomial $X_{j,1} + X_{j,2} T + \cdots + X_{j,n_j}
  T^{n_j-1} + T^{n_j}$ at $k+d_{1,j}+\cdots+d_{i-1,j}$. This
  polynomial has degree less than $n_j$, and we have to evaluate it at
  $\sum_{i=1,\dots,N} d_{i,j}$ points, so using fast multipoint
  evaluation~\cite[Chapter~10]{GaGe03}, this can be done in
  $O\tilde{~}(n_j+\sum_{i} d_{i,j})$ operations.

  Taking all $j$ into account, the overall time for evaluating these
  linear forms is $O\tilde{~}(N+\sum_{i,j} d_{i,j})$
  operations. Because for all $i=1,\dots,N$,
  $\sum_{j=1,\dots,m} d_{i,j}$ is at least equal to $1$ (otherwise, we
  would have $\dunder_i=(0,\dots,0)$), this is
  $O\tilde{~}(\sum_{1\leq i\leq N,1\leq j\leq m} d_{i,j})$. The cost
  needed to deduce all $g_i(\X)$ themselves is
  $O(\sum_{i,j} d_{i,j})$. This proves the second item.
  
  For the third point, remark first that the solutions of the system
  $\g=0$ are obtained by cancelling one factor in each $g_i$. For any
  given $j$ in $\{1,\dots,m\}$, our assumption on the characteristics
  of the base field implies that the affine forms
  $\kappa_{k+d_{1,j}+\cdots+d_{i-1,j}}(\X_j)$ showing up in the
  definition of $g_1,\dots,g_N$ are pairwise distinct, and thus (since
  they form a Vandermonde system) linearly independent. Thus, if we
  choose more than $n_j$ forms involving $\X_j$, we obtain an
  inconsistent linear system for $\X_j$. As a result, the solutions
  are obtained by choosing $n_1$ linear equations for $\X_1$, \dots,
  $n_m$ linear equations for $\X_m$. There are $\scrC_\bn(\d)$ such
  choices; for any of these choices, we recover the value of each
  $\X_j$ by solving a Vandermonde linear system; this can be done in
  quasi-linear time $O\tilde{~}(N)$~\cite[Chapter~10]{GaGe03}.

  Finally, to prove that all solutions are multiplicity-free, remark
  that locally around any of these solutions, the system is 
  equivalent to a linear system (since once we have chosen linear 
  equations to define the values of $\X_1,\dots,\X_m$, all other
  linear equations are non-zero).
\end{proof}


\subsection{The homotopy curve $\mathscr{Z}$} 

We now construct the homotopy itself. Given polynomials
$\f=(f_1,\dots,f_N)$ with multi-degrees
$\d=(\dunder_1,\dots,\dunder_N)$, with all $\dunder_i$ in $\N^m$, we
define the system $\g$ as above, together with the equations
$$\homotop(\f,\g, t) = t \f + (1-t) \g \in \K[t,\X],$$ for a new variable $t$. We
make the same assumption on the characteristics of the base field as
in the Lemma~\ref{lemma:g} (the assumptions on the $d_i$'s is
satisfied, since we assume that none of the $f_i$'s is constant).

Remark that $\homotop(\f,\g, 0)=\g$ and $\homotop(\f,\g, 1)=\f$.
Adding a new ``block'' of variables consisting only of $t$, the system
$\homotop(\f,\g, t)$ is seen to have multi-degree at most
$\d'=(d'_1,\dots,d'_N)$, with $d'_i = (1,d_{i,1},\dots,d_{i,m})$ for
all $i$; as said above, we correspondingly define
$\bn'=(1,n_1,\dots,n_m)$.

The system $\homotop(\f,\g, t)$ may not necessarily define a curve in
$\Kbar{}^{N+1}$ (for instance if $\f=-\g$, the fiber above $t=1/2$ has
dimension $N$). Let us then define the algebraic set $\mathscr{Z}$ as
the Zariski closure of $V(\homotop(\f,\g, t))-V(D)$, where $D$ is the
determinant of the Jacobian matrix $\mathbf{J}(\homotop(\f,\g, t))$ of
$\homotop(\f,\g, t)$ with respect to $\X_1,\dots,\X_m$. Finally, let
$\pi:\Kbar{}^{N+1} \to \Kbar$ denote the projection on the $t$-axis.

\begin{lemma}\label{lemma:degreeZ}
 The algebraic set $\mathscr{Z}$ has dimension one, the image by $\pi$
 of each of its irreducible components is dense, and it has degree at
 most $\scrC_{\bn'}(\d')$.
\end{lemma}
\begin{proof}
  The so-called Lazard Lemma~\cite[Proposition 3.4]{Morrison99}
  implies the dimension claims; as a result, we can apply
  Proposition~\ref{prop:degH} to obtain the degree bound.
\end{proof}

Let $\mathscr{I} \subset \K[t,\X]$ be the ideal
$\langle \homotop(\f,\g, t)\rangle : D^\infty$, so that $\mathscr{Z}$
is the zero-set of $\mathscr{I}$. Let us further denote by ${\frak I}$
the extension of $\mathscr{I}$ to $\K(t)[\X]$, and by
${\frak Z}\subset \overline{\K(t)}{}^N$ its zero-set; the Jacobian
criterion implies that ${\frak I}$ is radical, and that ${\frak Z}$
has dimension zero.  Let then $\lambda$ be a linear form with
coefficients in $\K$ that separates the points of ${\frak Z}$ (we will
discuss our choice for it further on). To $\lambda$, we can associate
a zero-dimensional parametrization $\scrQ=((q,v_1,\dots,v_N),\lambda)$
of ${\frak Z}$, where all polynomials have coefficients in
$\K(t)$. The previous lemma and Theorem~1 in~\cite{Schost03} imply the
following bound.

\begin{lemma}
  The numerator and denominator of all coefficients of all polynomials
  $q,v_1,\ldots,\penalty-10000v_N$ have degree at most $\scrC_{\bn'}(\d')$.
\end{lemma}

\subsection{Specialization properties} 
In our main algorithm, we use a classical tool, {\em lifting
  techniques}: to compute $\scrQ$, we compute the specialization of it
at $t=0$, lift it to a sufficient precision in $t$, and recover
$\scrQ$. Once we know $\scrQ$, we want to let $t=1$ in it, in order to
obtain a zero-dimenzional parametrization for $Z(\f)$.
In this paragraph, we give properties that underlie this process. 
First, we describe the situation at $t=0$.

\begin{lemma}\label{lemma:spec0}
  If a linear form $\lambda$ with coefficients in $\K$ is a separating
  element for $Z(\g)$, it is separating for ${\frak Z}$. When it is
  the case, $t$ divides no denominator in the corresponding
  zero-dimensional parametrization $\scrQ=((q,v_1,\dots,v_N),\lambda)$
  of ${\frak Z}$, and letting $t=0$ in these polynomials yields a
  zero-dimensional parametrization of $Z(\g)$.
\end{lemma}
\begin{proof}
  Consider the power series in $\K[[t]]$ obtained by lifting the
  points of $Z(\g)$ to solutions of $\homotop(\f,\g, t)$ using Newton
  iteration; call them $\Gamma_1,\dots,\Gamma_c$, with all $\Gamma_i$
  in $\K[[t]]^N \subset \Kbar[[t]]^N$ and $c=\scrC_{\bn}(\d)$. In the
  sequel, $\K((t))$ denotes the field of fractions of $\K((t))$.

  Because there are $c=\scrC_{\bn}(\d)$ such solutions, and ${\frak
    I}$ can have at most $c$ solutions (Proposition~\ref{prop:degH}),
  these power series are the {\em only} solutions of the extension of
  ${\frak I}$ to $\Kbar((t))[\X]$. The following well-known interpolation
  formulas
  \begin{equation}\label{eq:qv}
  q = \prod_{\bx \in {\frak Z}}(T-\lambda(\bx)), \quad
  v_i = \sum_{\bx=(x_1,\dots,x_N) \in {\frak Z}} x_i \prod_{\bx' \in {\frak Z},\ \bx' \ne \bx} (T-\lambda(\bx')) \quad (1 \le i \le N).
\end{equation}
define $\scrQ$; they show that all polynomials
  $q$ and $v_1,\dots,v_N$ have non-negative valuation at $t=0$
  and prove our claims.
\end{proof}

The situation at $t=1$ is more complex, since $\f$ may have fewer than
$\scrC_\bn(\d)$ roots. To state the relevant construction, we will
need power series centered at $t=1$ (and generalizations
thereof). Thus, we let $\tau=t-1$, and work with polynomials and power
series written in $\tau$ (the system $\homotop(\f,\g, t)$ written in
terms of $\tau$ becomes $\homotop(\f,\g, \tau)$). Let
$\varphi_1,\dots,\varphi_s$ be the points in $Z(\f)$; they belong to
$\Kbar{}^N$. Because the Jacobian matrix of $\f$ is invertible at
these points, we can use Newton iteration to lift them to power series
$\Phi_1,\dots,\Phi_s$ in $\Kbar[[\tau]]^N$ that cancel
$\homotop(\f,\g, \tau)$.

We will in fact need to describe all solutions of
$\homotop(\f,\g, \tau)$; for this, we use a slight generalization of
the presentation in~\cite{RRS}. That paper describes such solutions in
characteristic zero, where this is done by means of Puiseux series; in
arbitrary characteristic, this is not enough, so we will rely on the
fact that the ring $\L$ of all ``generalized power series''
$F=\sum_{i \in I} f_i {\tau}^i$, where the index set $I \subset \Q$
(that depends on $F$) is well-ordered and all $f_i$'s are in $\Kbar$,
contains an algebraic closure of $\Kbar((\tau))$~\cite{Rayner68}.

Because the exponent support is well-ordered, we can define the {\em
  valuation} of such a (non-zero) $F$ as the rational
$\nu(F)=\min(i \in I,\ f_i \ne 0)$; this extends the $\tau$-adic
valuation on $\Kbar((\tau))$. For such an element $F$, if
$\nu(F) \ge 0$, we write $\ell_0(F)$ for the coefficient of $\tau^0$
in the expansion of $F$ (and we extend this notation to vectors).

We will ensure below that we can apply Lemma~\ref{lemma:spec0}; as a
consequence, $\homotop(\f,\g, \tau)$ has $c=\scrC_{\bn}(\d)$ pairwise
distinct roots in an algebraic closure of $\Kbar((\tau))$. These roots
can then be written as $\Phi_1,\dots,\Phi_c$, with all $\Phi_i$
in~$\L^N$; up to reordering them, we can assume that the first $s$ of
them are the power series $\Phi_1,\dots,\Phi_s$ defined previously.

\begin{lemma}\label{lemma:nonrep}
  Let $c'$ in $\{s,\dots,c\}$ be such that $\Phi_1,\dots,\Phi_{c'}$ have
  all their coordinates with non-negative valuations.  Define
  $\varphi_1,\dots,\varphi_{c'}$ as the vectors in $\Kbar{}^N$ obtained
  as $\varphi_i=\ell_0(\Phi_i)$ for all $i$. 
  Then, for $i=1,\dots,s$ and $i'=s+1,\dots,c'$, $\varphi_i
  \ne \varphi_{i'}$ holds.
\end{lemma}
\begin{proof}
  Take $i$ and $i'$ as above. By Newton iteration, we know that
  $\Phi_i$ is the unique vector of power series in $\Kbar[[\tau]]$
  that cancels $\homotop(\f,\g, \tau)$ and such that
  $\ell_0(\Phi_i)=\varphi_i$. Hence, the only case we have to exclude
  is $\Phi_{i'}$ being a vector in $\L^N-\Kbar[[\tau]]^N$ and with
  $\ell_0(\Phi_{i'})=\varphi_i$.

  Suppose it is the case. By assumption, $\Phi_{i'}$ is not in
  $\Kbar[[\tau]]^N$, so one of its entries, say $\Phi_{i',j}$, is not
  in $\Kbar[[\tau]]$. The well-ordered nature of the exponent set of
  $\Phi_{i',j}$ shows that there exists $e$ in $\Q_{>0}$ such that
  $\tau^e$ is the smallest non-integer exponent appearing with
  non-zero coefficient in $\Phi_{i',j}$; if there are several such
  $j$'s, assume we have chosen one with smallest exponent $e$.

  Write $\Phi_{i'} =\Phi_{i',0} + \Phi_{i',1}$, where $\Phi_{i',0}$
  consists of all terms with exponent less than $e$; this is thus a
  vector of truncated power series, and all terms in $\Phi_{i',1}$
  have valuation at least $e$. Since
  $\homotop(\f,\g, \tau)(\Phi_{i'})=0$, Taylor expansion shows that
  $\homotop(\f,\g, \tau)(\Phi_{i',0}) + \mathbf{J}(\homotop(\f,\g,
  \tau))(\Phi_{i',0}) \Phi_{i',1} = O(\tau^{2 e})$,
  where the right-hand side consists of terms with valuation at least
  $2e$.  The invertibility of
  $\mathbf{J}(\homotop(\f,\g, \tau))(\varphi_{i})$ implies that the
  matrix $\mathbf{J}(\homotop(\f,\g, \tau))(\Phi_{i',0})$ is
  invertible too, so that
  \[\mathbf{J}(\homotop(\f,\g, \tau))(\Phi_{i',0})^{-1} \homotop(\f,\g,
  \tau)(\Phi_{i',0}) + \Phi_{i',1} = O(\tau^{2 e}).\]
  The first term is a power series, whereas by assumption
  $\Phi_{i',1}$ has at least one term with non-integer exponent
  $e$. This term cannot be cancelled by the right-hand side, a
  contradiction.
\end{proof}

Finally, in the discussion below, for $i=1,\dots,c$ and $j=1,\dots,N$,
we write $\mu_{i,j}=\nu(\Phi_{i,j})$ and $\mu_i=\min_{1 \le j \le N}
\mu_{i,j}$. In particular, $\mu_i \ge 0$ if and only if $i \le
c'$. Still inspired by~\cite{RRS}, we will say that a linear form
$\lambda$ with coefficients in $\K$ is a {\em well-separating} element
for $(\f,\g)$ if:
\begin{enumerate}
\item $\lambda$ is separating for $Z(\g)$
\item $\lambda$ is separating for $\{\varphi_1,\dots,\varphi_{c'}\}$
\item $\nu(\lambda(\Phi_i)) = \mu_i$ for all $i=1,\dots,c$.
\end{enumerate}
We will discuss later on how random choices can ensure these
properties with high probability. For the moment, remark that by
Lemma~\ref{lemma:spec0}, the first condition implies that 
$\lambda$ is separating for ${\frak Z}$.

Let us extend $\nu$ to $\L[T]$, by letting $\nu(a_0 + \cdots + a_s
T^s)=\min_{a_i \ne 0}(\nu(a_i))$. This applies in particular to
polynomials in $\K((\tau))[T]$; in that case, note that for any $f$ in
$\K(\tau)[T]$ and $e$ in $\Z$, $\tau^e f$ is in $\K[[\tau]][T]$ if and
only if $e + \nu(f) \ge 0$. This being said, we can state the main
result in this paragraph; it follows closely~\cite{RRS}, in our
slightly different setting.
\begin{lemma}\label{lemma:rw}
  Suppose that $\lambda$ is a well-separating element, let
  $\scrQ=((q,v_1,\dots,v_N),\lambda)$ be the corresponding
  zero-dimensional parametrization of ${\frak Z}$ over $\K((\tau))$,
  and let $e = -\nu(q)$.  Define the polynomials $q^\star=\tau^e q$
  and $(v_j^\star=\tau^e v_j)_{1 \le j \le N}$. Then, these
  polynomials are in $\K[[\tau]][T]$.

  Defining further $r_0$ as the leading coefficient of $q^\star(0,T)$ and
  $$r = \frac 1{r_0} q^\star(0,T) \quad\text{and}\quad w_j = \frac
  1{r_0} v^\star_j(0,T) \bmod r \quad (1 \le j \le N),$$
  the polynomials $r,w_1,\dots,w_N$ are such that
  $$r = \prod_{1\le i \le c'} (T-\lambda(\varphi_i)) \quad \text{and}\quad w_j = \sum_{1 \le i \le c'} \varphi_{i,j}
  \prod_{1 \le i' \le c',\ i' \ne i} (T-\lambda(\varphi_{i'})).$$
\end{lemma}
\begin{proof}
  To prove the first point, since all polynomials $v_j$ and $q$ have
  coefficients in $\K((\tau))$, it is enough to prove that $\nu(v_j)
  \ge \nu(q)$ holds for all $j$. In view of the interpolation formulas
  $$q = \prod_{1 \le i \le c} (T-\lambda(\Phi_i)), \quad v_j = \sum_{1
    \le i \le c} \Phi_{i,j} \prod_{1 \le i' \le c,\ i' \ne i}
  (T-\lambda(\Phi_{i'})),$$ we deduce first that $\nu(q) = \sum_{c' <
    i \le c} \mu_i$, and that for all $i,j$,
  $$\nu\left(\Phi_{i,j} \prod_{1 \le i' \le c,\ i' \ne i}
  (T-\lambda(\Phi_{i'}))\right) = \mu_{i,j} +
  \sum_{c' < i' \le c,\ i' \ne i}
  \mu_{i'} \ge \sum_{c' < i' \le c}
  \mu_{i'}=\nu(q).$$ Taking the sum, this implies that
  $\nu(v_j) \ge \nu(q)$, as claimed. Besides, since the definition of $e$ gives
  $e=-\sum_{c' < i \le c} \mu_i$, we obtain
  the factorization
  $$q^\star =\prod_{1 \le i \le c'} (T-\lambda(\Phi_i)) \cdot
  \prod_{c' < i \le c}
  (\tau^{-\mu_i}T-\tau^{-\mu_i}\lambda(\Phi_i)).$$
  In particular, the polynomial $r=q^\star(0,T)$ satisfies
  $$r = \gamma \prod_{1 \le i \le c'}
  (T-\ell_0(\lambda(\Phi_i))) = \gamma \prod_{1 \le i \le c'}
  (T-\lambda(\varphi_i)),$$ where $\gamma$ is the scalar
  $\gamma = \prod_{c' < i \le c}
  \ell_0(\tau^{-\mu_i}\lambda(\Phi_i))$; it is non-zero, as a consequence of the third
condition in the definition of a well-separating element. Proceeding similarly with $v_j^\star$, we obtain the
  claim for $w_j$.
\end{proof}

\subsection{Recovering $Z(\f)$} \label{ssec:recoverZ}
The polynomials $r$ and $w_1,\dots,w_N$ defined in the previous lemma
do not necessarily form a zero-dimensional parametrization of
$\{\varphi_1,\dots,\varphi_{c'}\}$, since $r$ may have multiple
roots. We show here how to deduce a zero-dimensional parametrization
of $Z(\f)=\{\varphi_1,\dots,\varphi_{s}\}$.

Our starting point is that the minimal polynomial in the
parametrization of $Z(\f)$ associated to $\lambda$ is $t=\prod_{1 \le
  i \le s}(T-\lambda(\varphi_i))$, and that this polynomial is a
factor of $r$.

More precisely, because $\lambda$ separates
$\{\varphi_1,\dots,\varphi_{c'}\}$, and because each $\varphi_i$, for $i$ in
$\{1,\dots,s\}$, only
appears once among $\varphi_1,\dots,\varphi_{c'}$
(Lemma~\ref{lemma:nonrep}), $\lambda(\varphi_i)$ is a root of $r$ of
multiplicity $1$, for all $i$ as above. Thus, we can assume without loss of
generality that the roots of $r$ of multiplicity $1$ are
$\lambda(\varphi_1),\dots,\lambda(\varphi_{c''})$, for some $c''$ in
$\{s,\dots,c'\}$, and let $r_1$ be the product $\prod_{1 \le i \le
  c''} (T-\lambda(\varphi_i))$, so that $t$ divides $r_1$. Explicitly,
we have (independently of the characteristic)
\begin{equation}\label{eq:r1}
r_1 = \frac{\tilde r}{\gcd(\tilde r, r')} \quad\text{with}\quad \tilde r= \frac r{\gcd(r, r')}.
\end{equation}
Let us write $r = r_1 r_{\ge 2}$, where $r_{\ge 2}$ is $\prod_{c'' < i
  < c'} (T-\lambda(\varphi(i))$, and define
$$y_i = \frac {w_i}{r_{\ge 2}} \bmod r_1, \quad 1 \le i \le N;$$
one easily sees that 
$$y_i = \sum_{1 \le i \le c''} \varphi_{i,j} \prod_{1 \le i' \le
  c'',\ i' \ne i} (T-\lambda(\varphi_i)).$$ In other words,
$((r_1,y_1,\dots,y_N),\lambda)$ is a zero-dimensional parametrization
of $\{\varphi_1,\dots,\varphi_{c''}\}$.

The set $\{\varphi_1,\dots,\varphi_{c''}\}$ contains
$Z(\f)=\{\varphi_1,\dots,\varphi_{s}\}$ and is contained in
$V(\f)$. To conclude, we remove from this set all points where the
Jacobian determinant of $\f$ vanishes. This is done as in Algorithm
${\sf Clean}$ of~\cite{GiLeSa01}, with one small modification: in that
result, zero-dimensional parametrizations did not involve rational
expressions of the roots of the form $x_i = v_i(T)/q'(T)$, but
polynomial ones of the form $x_i = v_i(T)$.  This is harmless, since
conversions between the two can be done in quasi-linear time.

\subsection{The algorithm and proof of
  Proposition~\ref{prop:homo1}} \label{ssec:algoK}
We can finally summarize the whole process in Algorithm~\ref{algo:1}
below. For the moment, we assume that a well-separating element
$\lambda$ is part of the input.

\begin{lemma}\label{lemma:10}
  Suppose that $\f=(f_1,\dots,f_N)$ has multi-degree at most
  $\d=(\dunder_1,\dots,\dunder_N)$, with all $\dunder_i$ in $\N^m$,
  and that $\f$ is given by a straight-line program $\Gamma$ of size
  $L$; suppose further that $\K$ has characteristic either zero or at
  least equal to $\max_{1 \le j \le m} d_{1,j} + \cdots +
  d_{N,j}$.
  Given $\Gamma$, $\d$ and a linear form $\lambda$ which is a
  well-separating element for $(\f,\g)$, one can compute a
  zero-dimensional parametrization of $Z(\f)$ using
  $$ O\tilde{~}\left (\scrC_\bn(\d)\scrC_{\bn'}(\d') \left (L+\sum_{i,j}d_{i,j}+N^2 \right)N \right)$$
  operations in $\K$. 
\end{lemma}
\begin{proof}
  The cost of Step~\ref{step:1:1} in {\sf NonsingularSolutions\_aux}
  follows from Lemma~\ref{lemma:g}.  Step~\ref{step:1:2} can be done
  in quasi-linear time $O\tilde{~}(\scrC_\bn(\d) N)$ using the
  algorithms of~\cite[Chapter~10]{GaGe03}.

  For the main step, computing the parametrization $\scrQ$ with
  coefficients in $\K(t)$, we use the lifting algorithm
  in~\cite{Schost03}. The main factor determining the cost of this
  algorithm is the required precision needed in $t$, that is, the
  degree of the coefficients in the output: Lemma~\ref{lemma:degreeZ}
  shows that it is at most $\scrC_{\bn'}(\d')$. The other important
  quantity is the size of the straight-line program that evaluates
  $\homotop(\f,\g, t)$: using Lemma~\ref{lemma:g}, we see that it is
  $O(L+\sum_{i,j} d_{i,j})$. We deduce that this step has cost
  $O\tilde{~}(\scrC_\bn(\d)\scrC_{\bn'}(\d')(L+\sum_{i,j}d_{i,j}+N^2)N)$.

  Step~\ref{step:1:4} involves exponent comparisons, setting some
  variable to zero and computing a remainder; it can be done in
  quasi-linear time $O\tilde{~}(\scrC_\bn(\d) N)$. Step~\ref{step:1:5}
  requires computing the polynomial $r_1$ using~\eqref{eq:r1}, and
  some computations modulo $r_1$; all of this can be done in time
  $O\tilde{~}(\scrC_\bn(\d) N)$.
  
  Finally, Step~\ref{step:1:6} takes $O(\scrC_\bn(\d) (L + N^2)N)$ to
  reduce $D$ modulo $(r_1,u_1,\dots,y_N)$, where the term $(L+N^2)N$
  is the size of the straight-line program that computes the Jacobian
  determinant $D$. The other operations at this stage take
  quasi-linear time $O(\scrC_\bn(\d) N)$.
\end{proof}
 
\begin{algorithm}
\begin{algorithmic}[1]
\Require $\Gamma$, $\d$, a well-separating element $\lambda$
\Ensure a zero-dimensional parametrization of $Z(\f)$
\State \label{step:1:1}Define $\g$ and compute $Z(\g)$ using Lemma~\ref{lemma:g}
\Statex {\sf Cost:} $O\tilde{~}(\scrC_\bn(\d) N)$
\State \label{step:1:2}Compute a zero-dimensional parametrization $\scrQ_\g$ for $Z(\g)$ using interpolation formulas~\eqref{eq:qv}
\Statex {\sf Cost:} $O\tilde{~}(\scrC_\bn(\d) N)$
\State \label{step:1:3}Apply the lifting algorithm of~\cite{Schost03} to $\scrQ_\g$ and $\homotop(\f,\g, t)$, to recover a zero-dimensional parametrization $\scrQ$ for ${\frak Z}$ with coefficients in $\K(t)$
\Statex {\sf Cost:} $O\tilde{~}(\scrC_\bn(\d)\scrC_{\bn'}(\d')(L+\sum_{i,j}d_{i,j}+N^2)N)$
\State \label{step:1:4}Compute $r$ and $w_1,\dots,w_N$ as in Lemma~\ref{lemma:rw}
\Statex {\sf Cost:} $O\tilde{~}(\scrC_\bn(\d) N)$
\State \label{step:1:5}Compute $r_1$ and $y_1,\dots,y_N$ as in Subsection~\ref{ssec:recoverZ}
\Statex {\sf Cost:} $O\tilde{~}(\scrC_\bn(\d) N)$
\State \label{step:1:6}Compute and return ${\sf Clean}(r_1,y_1,\dots,y_N,D)$
\Statex {\sf Cost:} $O\tilde{~}( \scrC_\bn(\d) (L+N^2)N)$
\caption{({\sf NonsingularSolutions\_aux}): Solving $\f$ by symbolic homotopy}
\label{algo:1}
\end{algorithmic}
\end{algorithm}


Our last question is how to ensure that with high probability, a
randomly chosen $\lambda$ is well-separating. For this, we can follow
the analysis of~\cite[Lemma~4.2]{RRS}: for a linear form $\lambda$ to
be well-separating, $\lambda$ must assume non-zero values on at most
$c^2$ non-zero vectors in $\Kbar{}^N$ (with $c=\scrC_\bn(\d)$), namely the
differences $\bx-\bx'$, for distinct $\bx,\bx'$ in
$Z(\g)$, the differences $\varphi_i-\varphi_{i'}$, for $i,i'$
in $\{1,\dots,c'\}$ such that $\varphi_i \ne \varphi_{i'}$, and the
coefficient vectors $({\rm coeff}(\Phi_{i,j}, \tau^{\mu_i}))_{1 \le j
  \le N}$, for $i$ in $\{1,\dots,c\}$.

The following classical result shows that a random choice of $\lambda$ is
well-separating with high probability, provided we pick it in a large
enough set.

\begin{lemma}\label{lemma:avoid}
  Let $\A$ be a domain containing a field $\K$, let
  $\bx_1,\dots,\bx_k$ be non-zero vectors in $\A^N$, and suppose that
  $\K$ has characteristic either zero or at least $8(N-1)k$. Consider
  the set of linear forms
  $$u^{(i)} = X_1 + i X_2 + \cdots + i^{N-1} X_N,$$ for $i$ in
  $\{1,\dots,8(N-1)k\}$. Then at least $7/8$ of these linear forms
  vanish on none of $\bx_1,\dots,\bx_k$.
\end{lemma}
\begin{proof}
  For any $i$ in $\{1,\dots,k\}$, write
  $\bx_i=(x_{i,1},\dots,x_{i,N})$ and consider the polynomial $P_i =
  x_{i,1} + x_{i,2} T + \cdots + x_{i,N} T^{N-1}$. This is a non-zero
  polynomial, so it has at most $N-1$ roots, and thus at most $N-1$
  roots in $\{1,\dots,8(N-1)k\}$. Taking all $i$'s into account, we
  see that at least $7/8$ of the elements in $\{1,\dots,8(N-1)k\}$
  cancel none of the polynomials $P_i$.
\end{proof}

We can then state the main algorithm of this section, together with
its probablity analysis (obviously, the cost is the same as that of
{\sf NonsingularSolutions\_aux}), which will finish the proof of
Proposition~\ref{prop:homo1}. 

\begin{algorithm}
\begin{algorithmic}[1]
\Require $\Gamma$, $\d$ \Ensure a zero-dimensional parametrization of
$Z(\f)$ \State \label{step:2:1} Set $\lambda=u^{(i)}$, for a randomly
chosen $i$ in $\{1,\dots,8(N-1)\scrC_\bn(\d)^2\}$
\State \label{step:2:} Return {\sf
  NonsingularSolutions\_aux}$(\Gamma,\d,\lambda)$
\caption{({\sf NonsingularSolutions}): Solving $\f$ by symbolic homotopy}
\label{algo:2}
\end{algorithmic}
\end{algorithm}


\begin{remark}\label{rem:failure}
We do not know how to verify the output of our algorithm with an
admissible cost (that is, similar to the cost of running the algorithm
itself). In any case, the output is a subset of $Z(\f)$; this is
ensured by our call to {\sf Clean} in the last step of {\sf
  NonsingularSolutions\_aux}. However, we may miss some solutions.

More precisely, if $\lambda$ is a well-separating element, which
occurs with probability at least $7/8$, the output is $Z(\f)$ itself;
otherwise we may obtain a subset of it, or {\sf fail}, when for
instance the assumptions of Lemma~\ref{lemma:rw} are not
satisfied (this analysis establishes Proposition~\ref{prop:homo1}).

Running the algorithm $k$ times, we obtain $k$ outputs, and a
zero-dimensional parametri\-zation of $Z(\f)$ lies among these $k$
outputs with probability at least $1-1/8^k$. If it is the case, since
all other outputs have degree less than that of $Z(\f)$, 
the correct outputs are the ones with highest degree. 
\end{remark}


\subsection{Example}

To illustrate the algorithm, let us consider a simple example. In this
example, we work over $\Q$ (later on, the algorithm of this section
will be applied over a prime field, but it is of course valid over
$\Q$ as well). We take $m=2$ and $\bn=(1,2)$, so that $N=3$, and that
our variables are $\X_1=(X_{1,1})$ and $\X_2=(X_{2,1},X_{2,2})$.  We
take polynomials $\f=(f_1,f_2,f_3)$ having respective multi-degrees
$\d=(\dunder_1,\dunder_2,\dunder_3)$, with
$\dunder_1=\dunder_2=\dunder_3=(1,1)$ (that is, they are
bilinear). Explicitly,
\begin{align*}
f_1&=      -16X_{1,1}X_{2,1} + 8X_{1,1},\\
f_2&=      -8X_{1,1}X_{2,1} - 16X_{1,1}X_{2,2} - 4X_{1,1},\\
f_3&=      3X_{1,1}X_{2,1} + 4X_{1,1}X_{2,2} + X_{1,1} + 2X_{2,1} + 4.
\end{align*}
The quantity $\scrC_\bn(\d)$ is the coefficient of
$\vartheta_1\vartheta_2^2$ in $(\vartheta_1 + \vartheta_2)^3 \bmod
\langle \vartheta_1^2,\vartheta_2^3\rangle$, that is, $3$;
similarly, with $\d'=(\dunder'_1,\dunder'_2,\dunder'_3)$,
all $\dunder'_i$ being equal to $(1,1,1)$, and $\bn'=(1,1,1)$,
we see that  $\scrC_{\bn'}(\d')$ is the sum of the coefficients of
$(\vartheta_0 + \vartheta_1 + \vartheta_3)^3 \bmod
\langle \vartheta_0^2,\vartheta_1^2,\vartheta_2^3\rangle$, that is, $12$.

The system $\g$ is given by
\begin{align*}
g_1&= X_{1,1}X_{2,1}\\
g_2&= (X_{1,1} + 1)(X_{2,1} + X_{2,2} + 1)\\
g_3&= (X_{1,1} + 2)(X_{2,1} + 2X_{2,2} + 4),
\end{align*}
its solutions being $(-2,0,-1),(-1,0,-2),(0,2,-3)$ (so it has
$\scrC_\bn(\d)=3$ solutions, as claimed). Using
$\lambda=X_{1,1}+2X_{2,1}+4X_{2,2}$, the corresponding zero-dimensional 
parametrization is
$$\scrQ_\g=(( T^3 + 23T^2 + 174T + 432,\ -3T^2 - 48T - 192,\ 2T^2 +
30T + 108,\ -6T^2 - 90T - 330), \lambda).$$ Applying Newton iteration,
we deduce a zero-dimensional parametrization  with coefficients in $\Q(t)$
of the form
$\scrQ=((q,v_1,v_2,v_3),\lambda)$ that
describes the solutions of $t\f+(1-t)\g$; the coefficients that appear
have numerators and denominators of degree at most $8 \le
\scrC_{\bn'}(\d')=12$. For instance, the first two terms of $q$
are
{\tiny
\[
T^3 +\frac{ 9561314t^7 - 35955867t^6 + 43077203t^5 - 18750948t^4 + 2544440t^3 - 152707t^2 + 4291t - 46}
{1081710t^7 - 3054661t^6 + 2913623t^5 - 1066868t^4 + 133524t^3 - 7525t^2 + 199t - 2} T^2
+ \cdots,
\]
}
where as a sanity check, we can verify that letting $t=0$ gives back the polynomial $T^3+23 T^2+\cdots$
that we started from.

The $(t-1)$-adic valuation of $q$ is $-1$, which means that the
integer $e$ of Lemma~\ref{lemma:rw} is $1$. Hence, we multiply $q$ and
$v_1,v_2,v_3$ by $(t-1)$, to obtain polynomials
$q^\star,v^\star_1,v^\star_2,v^\star_3$, in which we can evaluate $t$
at $1$. In particular, we obtain $q^\star(1,T)=-80/17T^2 - 880/17T$,
whose leading coefficient is $r_0=-80/17$ (remark that Lemma~\ref{lemma:rw}
uses evaluation at $0$, since we work in variable $\tau=t-1$ in that paragraph).
 Still following
Lemma~\ref{lemma:rw}, we can then define $r=1/r_0\, q^\star(1,T)$ and
$w_j=1/r_0\, v_j^\star(1,T) \bmod r$, for $j=1,2,3$; explicitly, they are given
by
$$r=T^2+11T,\ w_1=-10T,\ w_2=-\frac 32T - 22,\ w_3= \frac 12T + 11.$$
The roots of $r$ are $T=0$ and $T=-11$ (both with multiplicity $1$, so 
we do not need to clean multiple roots); evaluating 
$(w_1/r', w_2/r', w_3/r')$ at $T=0$ and $T=-11$, we find the points
$\bx=( -10, 1/2, -1/2)$ and $\bx'=( 0, -2, 1 )$.

Both cancel $\f=(f_1,f_2,f_3)$; on the other hand, the Jacobian determinant
of $\f$ vanishes at $\bx'$, but not at $\bx$. We can then 
conclude that $Z(\f)=\{\bx\}$.


\section{The main algorithm: proof of Theorem~\ref{thm:homoZ}}\label{sec:bounds}

In this section, we work over $\K=\Q$ and we use the bounds on the
height of polynomials appearing in a zero-dimensional parametrization
of a set $Z(\f)$ given before in the context of a lifting algorithm
following that of~\cite{GiLeSa01}. 



\subsection{The lifting algorithm}\label{sec:mainalgo}

Our goal is now to give boolean complexity statements for the
computation of a zero-dimensional representation of $Z(\f)$. Given a
well-chosen prime $p$, we start by computing a zero-dimensional
parametrization of $Z(\f \bmod p)$, which is then lifted to a
zero-dimensional paramet\-rization of $Z(\f)$.

Recall that we assume that we are given an oracle $\mathscr{O}$, which
takes as input an integer $B$, and returns a prime number in
$\{B+1,\dots,2B\}$, uniformly distributed within the set of primes in
this interval~\cite[Section~18.4]{GaGe03}. Recall as well the
statement of Theorem~\ref{thm:homoZ}.

{\em  
  Suppose that $\f=(f_1,\dots,f_N)$ satisfies $\mdeg(\f) \le
  \d=(\dunder_1,\dots,\dunder_m)$ and $\hgt(\f)\le \bs=(s_1,\dots,s_N)$, and that $\f$
  is given by means of a straight-line program $\Gamma$ of size $L$,
  that uses integer constants of height at most $b$.

  There exists an algorithm {\sf NonsingularSolutionsOverZ} that takes
  $\Gamma$, $\d$ and $\bs$ as input, and that produces one of the
  following outputs:

\smallskip

  \begin{itemize}
  \item either a zero-dimensional parametrization of $Z(\f)$,

\smallskip

  \item or a zero-dimensional parametrization of degree less than that of $Z(\f)$,

\smallskip
  \item or {\sf fail}.
  \end{itemize}

\smallskip

\noindent 
  The first outcome occurs with probability at least $21/32$. In any
  case, the algorithm uses
  $$ O\tilde{~}\left (L b + \scrC_\bn(\d)\scrH_{\bn}(\bbeta,\d) \left (L+ N\mydeg +N^2 \right) N(\log(\myheight) + N )
  \right)$$
  boolean operations, with
  $$\mydeg = \max_{1 \le i \le N} d_{i,1} + \cdots + d_{i,m},\quad
\myheight =\max_{1 \le i \le N}(s_i) \quad\text{and}\quad
\bbeta=\left (s_i + \sum_{j=1}^m \log(n_j+1) d_{i,j}\right)_{1 \le i \le N}.$$
The
  algorithm calls the oracle $\mathscr{O}$ with an input parameter
  $B=\myheight \mydeg^{O(N)}$ and the polynomials in the output have
  degree at most $\scrC_\bn(\d)$ and height
  $O\tilde{~}(\scrH_{\bn}(\bbeta,\d) + N \scrC_\bn(\d))$.
}

As in the case of Proposition~\ref{prop:homo1}, running the algorithm
$k$ times gives a list of outputs among which is at least one
zero-dimensional parametrization of $Z(\f)$ with probability at least
$1-(11/32)^k$; observe also that all incorrect answers have degree less
than that of $Z(\f)$. 

The input size of the algorithm is $O(Lb)$ bits, whereas the output
size is \sloppy
$O\tilde{~}(N \scrC_\bn(\d) (\scrH_{\bn}(\bbeta,\d) + N
\scrC_\bn(\d)))$
bits; thus, up to polynomial factors in $N,d,\log(\myheight),L$, the
cost of the algorithm is close to our upper bound on the combined size
of its input and output. We are not aware of previous results that
would take multi-homogeneous bit-size bounds into account in such a
manner.

In order to quantify primes of ``bad reduction'', we need to introduce
several quantities related to $Z(\f)$. In addition to $\d$, $\bbeta$
and $s$ as given above, we define

\smallskip

\begin{itemize}
\item $\mu_1= N \log(8N \scrC_\bn(\d)^2)$, \smallskip
\item $\mu_2 = \scrH_\bn(\bbeta,\d) + 2\log(N + 1) \scrC_\bn(\d),$ \smallskip
\item $\mu_3 = \mu_2 + \mu_1 \scrC_\bn(\d)  +\log(N+2) \scrC_\bn(\d) + (N+1) \log(\scrC_\bn(\d))$,\smallskip
\item $H = 6N(\mydeg+1)\scrC_\bn(\d) \left ( \mu_3 + \myheight +\log(N+1) \scrC_\bn(\d)  \right )$,\smallskip
\item $H' = \scrH_\bn(\bbeta,\d) + (\mu_1+4 \log(N+2))\scrC_\bn(\d),$\smallskip
\item $e = \max_{1 \le j \le m} d_{1,j} + \cdots + d_{N,j}$, \smallskip
\item $B =\max\left ( 8 \lceil  H \rceil, e \right )$.\smallskip
\end{itemize}

Here is how these quantities come into play. We will run Algorithm
{\sf NonsingularSolutions} with input $\f \bmod p$, for a prime
$p$. The separating element used in this algorithm has coefficients in
$\F_p$; once lifted back to $\Z$ in the canonical manner, the
construction used in that algorithm shows that it has height at most
$\mu_1$.

Next, using Lemmas~8 and~9 in~\cite{DaMoScWuXi05}, we deduce that
there is a positive integer $A$ such that we have

\smallskip

\begin{itemize}
\item $\log(A) \le H$
\smallskip
\item for any prime $p$ that does not divide $A$, $Z(\f)$ and
  $Z(\f \bmod p)$ have the same cardinality.
\end{itemize}

\smallskip

\noindent Now, remark the following:

\smallskip

\begin{itemize}
\item there are at least $B/2\log(B)$ primes in $\{B+1,\dots,2B\}$, by~\cite[Theorem~18.8]{GaGe03};

\smallskip

\item there are at most $\log(A)/\log(B) \le H/\log(B)$ primes in  $\{B+1,\dots,2B\}$
  that divide~$A$.
\end{itemize}

\smallskip

Let $p$ be a prime in $\{B+1,\dots,2B\}$, which we obtain by calling
the oracle $\mathscr{O}$ with input parameter $B$. By the discussion
above, the probability that $p$ divides $A$ is at most $2 H/B$, which
is at most $1/4$ by construction; on the other hand, $B$ has been
chosen small enough to be $\myheight \mydeg^{O(N)}$, so that $\log(B)$
is $O(\log(\myheight) + N \log(\mydeg))$.

As in the algorithm in~\cite{GiLeSa01}, we start by solving the system
modulo $p$, then lift this solution to a zero-dimensional
parametrization of $Z(\f)$. By definition of $B$, the field $\F_p$
satisfies the assumptions of Proposition~\ref{prop:homo1}, since $B$
is at least $\max(e,8(N-1)\scrC_\bn(\d)^2)$.  Thus, we can call
Algorithm {\sf NonsingularSolutions}, with input the straight-line
program $\Gamma'$ obtained by reducing all constants appearing in
$\Gamma$ modulo $p$ (computing these constants takes time
$O\tilde{~}(L (\log(B)+ b))=O\tilde{~}(L (\log(\myheight)+N
\log(\mydeg)+ b))$).  Recall that we obtain

\smallskip

\begin{itemize}
\item either a zero-dimensional parametrization of $Z(\f \bmod p)$,
\smallskip
\item or a zero-dimensional parametrization of degree less than that of $Z(\f \bmod p)$,
\smallskip
\item or {\sf fail},
\end{itemize}

\smallskip

\noindent with the first outcome arising with probability at least $7/8$.  In
all cases, since operations modulo $p$ take
$O\tilde{~}(\log(\myheight)+N\log(\mydeg))$ bit operations, the running
time is
\begin{equation}\label{eq:rt1}
  O\tilde{~}\left (\scrC_\bn(\d)\scrC_{\bn'}(\d') \left (L+\sum_{1
        \le i \le N, 1 \le j \le m}d_{i,j}+N^2 \right) N (\log(\myheight) + N \log(\mydeg))
  \right)  
\end{equation}
bit operations. If this computation fails, our main
algorithm will return {\sf fail} as well. Else, we have obtained a
zero-dimensional parametrization $\scrQ_0=(( q_0, v_{1,0},\dots,
v_{N,0}), \lambda_0)$.

Let then $\lambda$ be the canonical lift of $\lambda_0$ to a linear
form with non-negative integer coefficients; as said previously, the
way $\lambda_0$ is chosen implies that $\lambda$ has height at most
$\mu_1=N \log(8 N \scrC_\bn(\d)^2)$.  Using Newton
iteration~\cite[Section~4.3]{GiLeSa01}, we deduce the existence of a
zero-dimensional parametrization $ \scrQ_\infty=((q_\infty,
v_{1,\infty},\dots, v_{N,\infty}), \lambda)$ with coefficients in the
$p$-adic integers $\Z_p$, that describes a subset of $Z(\f)$ over an
algebraic closure of the field of $p$-adic numbers $\Q_p$. We run the
lifting algorithm of~\cite[Section~4.3]{GiLeSa01} up to a precision at
least equal to $2H'$, from which we reconstruct a rational
parametrization with rational coefficients.

\smallskip

\begin{itemize}
\item Suppose that $Z(\f)$ and $Z(\f \bmod p)$ have the same
  cardinality, and that $\scrQ_0$ describes $Z(\f \bmod p)$; this is
  the case in particular when $p$ does not divide $A$, and $\scrQ_0$
  is a zero-dimensional parametrization of $Z(\f \bmod p)$, so it
  occurs with probability at least $7/8 \times 3/4 = 21/32$, as claimed.

  Then, by reasons of cardinality, the zero-dimensional parametrization
  $\scrQ_\infty$ actually describes all of $Z(\f)$, over an algebraic
  closure of $\Q_p$. Since $Z(\f)$ is defined over $\Q$, and since
  $\lambda$ has coefficients in $\Z$, we deduce that all coefficients in
  $\scrQ_\infty$ actually belong to $\Q$: indeed, these polynomials show
  up as a Gr\"obner basis in $\Q_p[X_1,\dots,X_N,T]$ of the ideal
  generated by the defining ideal of $Z(\f)$, together with $T-\lambda$.

  Since the separating element constructed by {\sf NonsingularSolutions} has
  coefficients of height  at most $\mu_1$,
 Proposition~\ref{prop:hgt} shows that all coefficients in
  $\scrQ_\infty$ are rational numbers of height at most $H'$. Hence, 
  knowing them modulo a number greater than $\exp(2H')$ is sufficient
  to reconstruct them.

\smallskip

\item Otherwise, either $Z(\f)$ and $Z(\f \bmod p)$ do not have the
  same cardinality, or $\scrQ_0$ describes a proper subset of
  $Z(\f \bmod p)$. Since the lifting argument above shows that
  $Z(\f \bmod p)$ must have cardinality at most equal to that of
  $Z(\f)$, in all cases, $\scrQ_0$ has degree less than that of
  $Z(\f)$, and similarly for the output of the lifting algorithm.
\end{itemize}

\smallskip

\noindent In any case, the dominant part of this process is lifting, since
reconstructing rational numbers from their $p$-adic expansion can be
done in quasi-linear time~\cite[Chapter~11]{GaGe03}. Using the cost analysis
from~\cite{GiLeSa01}, we deduce that the cost is
\begin{equation}\label{eq:rt2}
 O\tilde{~}\left (\scrC_\bn(\d) H'  \left (L+N^2 \right) N \right)  
\end{equation}
bit operations. 

Up to logarithmic factors, the height bound $H'$ on the output is
$O\tilde{~}(\scrH_{\bn}(\bbeta,\d) + N \scrC_\bn(\d))$.  Remark now
that the definitions of $\scrH_\bn(\bbeta,\d)$ and $\scrC_{\bn'}(\d')$
are very similar, and imply that we have
$\scrC_{\bn}(\d) \le \scrC_{\bn'}(\d') \le \scrH_\bn(\bbeta,\d)$.
Thus, we deduce from~\eqref{eq:rt1} and~\eqref{eq:rt2} the following
upper bound on the total boolean cost of our algorithm:
$$ O\tilde{~}\left (L b + \scrC_\bn(\d)\scrH_{\bn}(\bbeta,\d) \left (L+ N\mydeg +N^2 \right) N(\log(\myheight) + N\log(\mydeg) )
\right).$$



\section{Application to polynomial minimization}\label{sec:appli}

We finally turn to the last question mentioned in the introduction:
given polynomials $\h=(h_1, \ldots, h_p)\subset \Z[X_1, \ldots, X_n]$,
that define an algebraic set $V=V(\h)\subset \C^n$, determine
$\min_{\bx \in V \cap \R^n} \pi_1(\bx)$, where $\pi_1$ is the
canonical projection $ (x_1, \ldots, x_n)\mapsto x_1$.

Our goal is to give boolean complexity estimates for the computation
of this minimum, under some genericity assumptions on $\h$. The
assumptions on $\h$ are discussed in the first subsection, which also
contains the statement of the main result of this section
(Theorem~\ref{thm:main}). Next, we discuss the Lagrangian
reformulation of our minimization problem; this allows us prove
Theorem~\ref{thm:main} in the last subsection.


\subsection{Genericity assumptions}\label{ssec:gen}

Let $\h=(h_1, \ldots, h_p)$ be our input polynomials and let
$V\subset \C^n$ be their zero-set. In general, in cases where we may
not necessarily assume $V$ smooth, the {\em critical points} of
$\pi_1$ on $V$ are those points $\bx \in V$ that do not belong to the
singular locus of $V$ and at which $T_\bx V$ is ``vertical'', in the
sense that $\pi_1(T_\bx V)= \{0\}$; following~\cite{BaGiHeMb01,
  BaGiHePa05}, we denote this set by $W(\pi_1,V)$.

Let $\jac(\h)$ be the Jacobian matrix of
$\h$ and let $\jac(\h, 1)$ denote the truncated
jacobian matrix (which in general is rectangular)
\[
\left [\begin{array}{ccc}
  \frac{\partial h_1}{\partial X_2}& \cdots\cdots & \frac{\partial h_1}{\partial X_n} \\
\vdots & & \vdots \\
  \frac{\partial h_p}{\partial X_2}& \cdots\cdots & \frac{\partial h_p}{\partial X_n} \\
\end{array}\right ].
\]
Following again the construction of~\cite{BaGiHePa05}, if $m$ is a
$(p-1)$-minor of $\jac(\h,1)$, $ {\sf Minors}(\h, m)$ denotes the
vector of $p$-minors of $\jac(\h, 1)$ obtained by adding the missing
row and the missing column to $m$; there are $n-p$ such minors. Then,
we say that $(h_1, \ldots, h_p)$ satisfies assumption $\GS$ if the
following conditions hold:
\begin{itemize}
\item[(1)] At any point of $V$, the jacobian matrix $\jac(\h)$ has full rank
  $p$.

  This implies that if not empty, $V$ is smooth and
  $(n-p)$-equidimensional and $\h$ generates its vanishing ideal.  As
  a further consequence, the set $W(\pi_1,V)$ of critical points of
  $\pi_1$ on $V$ consists exactly of those points $\x$ that satisfy
  the conditions
  \[h_1(\bx)=\cdots =h_p(\bx)=0, \qquad \rank(\jac_\bx(\h, 1))\leq p-1,\]
  and the minimizers of $\pi_1$ on $V \cap \R^n$ form a subset of $W(\pi_1,V)$.

\item[(2)] The truncated jacobian matrix $ \jac(\h, 1)$ has rank $p-1$
  at all $\bx\in W(\pi_1, V)$.

\item[(3)] The set $W(\pi_1, V)$ is finite.

\item[(4)] For any $(p-1)$-minor $m$ of $\jac(\h,1)$, the polynomials
  $\h, {\sf Minors}(\h, m)$ define $W(\pi_1, V)$ in the Zariski open
  set ${\cal O}(m)$ defined by $m\neq 0$ and their jacobian matrix has
  full rank $n$ at any point of $W(\pi_1, V)\cap {\cal O}(m)$.
\end{itemize}

\smallskip\noindent
We can now state the main result of this section. 

\begin{theorem}\label{thm:main}
  Let $\h=(h_1, \ldots, h_p)\subset \Z[X_1, \ldots, X_n]$, and assume
  that all $h_i$'s have degree at most $\mydegh$ and height at most
  $\myheight$. Assume further that $\h$ satisfies $\GS$ and is given by a
  straight-line program $\Gamma$ of length $E$, that uses integers of
  height at most $\myheight'$.

  Then, there exists a randomized algorithm that takes $\Gamma$,
  $\mydegh$ and $\myheight$ as input, and computes a zero-dimensional
  parametrization of the set of critical points of $\pi_1$ on $V(\h)$
  with probability at least $147/256\geq 0.57$ and using
$$ O\tilde{~}\left ( p (E+n) \myheight' + n^{3}
  {{n-1} \choose {p-1}}{n \choose p}
  (\myheight+\mydegh)\mydegh^{2p}(\mydegh-1)^{2(n-p)}(pE+n\mydegh+n^2)
\right ).$$
boolean operations. Moreover, the output polynomials have degree at
most ${{n-1}\choose{p-1}}\mydegh^p(\mydegh-1)^{n-p}$ and height
$O\tilde{~}
(n{{n}\choose{p}}(\myheight+\mydegh)\mydegh^p(\mydegh-1)^{n-p} )$.
\end{theorem}

We now prove that assumption $\GS$ is generic.  The proof of this
proposition occupies the rest of this subsection.  In what follows, we
let $\C[X_1, \ldots, X_n]_\mydegh$ denote the subset of polynomials in
$\C[X_1, \ldots, X_n]$ of degree at most $\mydegh$; we can see this as
an affine space of dimension ${{n+\mydegh} \choose n}$.


\begin{proposition}\label{prop:genericity}
  Let $\mydegh$ be a positive integer. There exists a
  nonempty Zariski open set $\mathscr{O}\subset \C[X_1, \ldots,
  X_n]_\mydegh^p$ such that any $\h\in \mathscr{O}$ satisfies $\GS$.
\end{proposition}

Let ${\cal N}={{n+\mydegh}\choose{n}}$ be the number of monomials of
degree at most $ \mydegh$ in $\C[X_1, \ldots, X_n]$ and denote these
monomials by $1={\frak m}_1, \dots, {\frak m}_{\cal N}$; they form a
$\C$-vector space basis of $\C[X_1, \ldots, X_n]_\mydegh$.  For
$1\leq i \leq p$, denote by ${\frak h}_i$ the polynomial
$\sum_{j=1}^{\cal N}\gamma_{i,j} {\frak m}_j$, where the
$\gamma_{i,j}$'s are new indeterminates, and by $\K$ the field of
rational fractions $\C(\gamma_{1,1}, \ldots, \gamma_{p, {\cal N}})$.
We consider the sequence
${\frak H}=({\frak h}_1,\ldots, {\frak h}_p)$; it is seen as a
sequence of polynomials in $\K[X_1, \ldots, X_n]$.

Polynomials in $\C[X_1, \ldots, X_n]_\mydegh$ are obtained by
instantiating the indeterminates $\gamma_{i,j}$ to elements of $\C$,
so we can we identify a polynomial $f$ with the sequence of coefficients
of ${\frak m}_1, \dots, {\frak m}_{\cal N}$ in it. In a similar way, a sequence of
polynomials in $\C[X_1, \ldots, X_n]_\mydegh^p$ is identified with
elements of $\C^{{\cal N}p}$ and, by abuse of notation, given a subset
$A\subset \C^{{\cal N}p}$ we may use the notation
``$\h=(h_1, \ldots, h_p) \in A$'' to denote a family of polynomials in
$\C[X_1, \ldots, X_n]_\mydegh^p$ whose sequence of coefficients belongs
to $A$.


\paragraph*{Genericity of $\GS(1)$.}
We first prove that for a generic choice of $\h$, at any point of
$V(\h)$, the jacobian matrix $\jac(\h)$ of $\h$ has full rank $p$. 
In this paragraph, we consider the polynomials ${\frak l}_i={\frak
  h}_i-\gamma_{i,1}$ for $1\leq i \leq p$; hence ${\frak l}_i$ has no
constant term, and belongs to $\K'[X_1, \ldots, X_n]$, where
$\K'\subset \K$ is the field of rational fractions
$\C((\gamma_{i,j})_{1\leq i \leq p, 2\leq j \leq {\cal N}})$. Let $\psi$
denote the mapping
$$\begin{array}{r@{~}c@{~}c@{~}c}
  \psi :& \overline{\K'}^n &\longrightarrow&\overline{\K'}^p\\
  & \mathbf c&\longmapsto& ({\frak l}_1(\mathbf c),\ldots,{\frak l}_p(\mathbf c)).
\end{array}$$
Let $K_0\subset \overline{\K'}^p$ be the set of critical values of
$\psi$.  By Sard's Theorem \cite[Chap. 2, Sec. 6.2, Thm
  2]{Shafarevich77}, $K_0$ is contained in a proper closed subset of
the closure of the image of $\psi$, and thus of $\overline{\K'}^p$.

We use $\gamma_{1,1},\dots,\gamma_{p,1}$ as coordinates in the
target space. Then, the ideal of
$\K'[\X,\gamma_{1,1},\dots,\gamma_{p,1}]$ generated by
${\frak l}_1+\gamma_{1,1},\dots,{\frak l}_p+\gamma_{p,1}$ and the
maximal minors of $\jac({\frak l}_1,\dots,{\frak l}_p)$ contains a
non-zero polynomial $P \in \K'[\gamma_{1,1},\dots,\gamma_{p,1}]$.  Up
to multiplying $P$ by a suitable denominator, we can then assume that
$P$ lies in $\C[ (\gamma_{i,j})_{1\leq i \leq p, 1\leq j \leq {\cal N}}]$ and
belongs to the ideal generated by the above polynomials in
$\C[ (\gamma_{i,j})_{1\leq i \leq p, 1\leq j \leq {\cal N}},\X]$.

Remark now that the generators we consider can be rewritten as
${\frak h}_1,\dots,{\frak h}_p$ and the maximal minors of
$\jac({\frak h}_1,\dots,{\frak h}_p)$. Thus, if we define
$\mathscr{O}_1\subset \C^{{\cal N}p}$ as the non-empty Zariski open
$\C^{{\cal N}p}-V(P)$, we deduce that for any $\h\in \mathscr{O}_1$, $\GS(1)$
holds.

\paragraph*{Genericity of $\GS(2)$.}
For the remaining genericity properties, we will use the fact that for
any system $\h$ that satisfies $\GS(1)$, these properties are known to
hold in generic coordinates. From this, we will deduce our claims 
using several times the following arguments.

Let ${\frak A}$ be the $n\times n$ matrix $\left (\alpha_{k,
  \ell} \right )_{1\leq k, \ell\leq n}$, where the $\alpha_{k,\ell}$'s
are new indeterminates.  We denote by $\FF$ the field of rational
fractions in the indeterminates $\gamma_{i,j}$ and $\alpha_{k,\ell}$ (for
$1\leq i \leq p$, $1\leq j \leq {\cal N}$ and $1\leq k,\ell\leq n$) with
coefficients in $\C$; we will also consider its subfield
$\FF'=\C(\alpha_{1,1}, \ldots, \alpha_{n,n})$.  For $f\in \FF[X_1,
  \ldots, X_n]$, we denote by $f^{\frak A}$ the polynomial $f({\frak
  A} \X)$; for a subset $F\subset \FF[X_1, \ldots, X_n]$, $F^{\frak
  A}$ denotes the set $\{f^{\frak A}\mid f\in F\}$. These notations
are naturally extended to the situation where we let a matrix $\mA\in
\GL_n(\C)$ act on $(X_1, \ldots, X_n)$.

We prove here that for a generic choice of $\h$, the matrix
$\jac(\h,1)$ has rank at least $p-1$ at any $\bx$ in $V(\h)$; this
will prove that it has rank exactly $p-1$ at the points of
$W(\pi_1,V(\h))$.

Let $\Delta({\frak H}, {\frak A})$ be the vector of $(p-1)$-minors of
$\jac({\frak H}^{\frak A}, 1)$ and
${\frak S}({\frak H}, {\frak A})\subset \F[X_1, \ldots, X_n]$ be the
polynomial sequence
\[{\frak H}^{\frak A}, \Delta({\frak H}, {\frak A});\]
remark that the polynomials $\Delta({\frak H}, {\frak A})$ are {\em
  not} obtained by applying the change of variables ${\frak A}$ to the
$(p-1)$-minors of $\jac({\frak H}, 1)$.  For $\h\in \C^{{\cal N}p}$ and
$\mA\in \GL_n(\C)$, we denote by
${\frak S}(\h, {\frak A})\subset \F'[X_1, \ldots, X_n]$,
${\frak S}({\frak H}, \mA)\subset \K[X_1, \ldots, X_n]$ and
${\frak S}(\h, \mA)\subset \C[X_1, \ldots, X_n]$ the polynomial
sequences obtained by instantiating ${\frak H}$ to $\h$ and/or
${\frak A}$ to $\mA$.


Let $r$ be the dimension of the zero-set of ${\frak S}({\frak H},
{\frak A})$ over an algebraic closure of $\F$. We first prove that
this dimension is $-1$.

Indeed, there exists a non-zero polynomial $\Lambda$ in
$\C[(\gamma_{i,j})_{1\leq i \leq p, 1\leq j \leq {\cal
    N}},(\alpha_{k,\ell})_{ 1\leq k,\ell\leq n}]$
such that for any $\h,\mA$ that do not cancel $\Lambda$, the zero-set
of the system ${\frak S}(\h, \mA)$ has dimension $r$ as well.  Fix
$\h$ such that $\Lambda(\h,{\frak A})$ is not zero and such that $\h$
belongs to $\mathscr{O}_1$ (such an $\h$ exists). Since $\h$ then
satisfies $\GS(1)$, using the third item in~\cite[Proposition B.1
(elec. appendix)]{SaSc13}, we deduce that there exists a non-empty
Zariski open set $\mathscr{A}_\h$ of $\C^{n\times n}$ such that for
$\mA\in \mathscr{A}_\h$, the zero-set of ${\frak S}(\h, \mA)$ has
dimension $-1$. On the other hand, by assumption on $\h$, for a
generic $\mA$, the value $\Lambda(\h,\mA)$ is not zero; in that case,
the zero-set of ${\frak S}(\h, \mA)$ has dimension $r$. Thus, our
claim $r=-1$ is proved.

Repeating the specializing argument, but with respect to the variables
$\alpha_{k,\ell}$, we choose $\mA\in \mathscr{A}$ such that
$\Lambda_\mA=\Lambda((\gamma_{i,j})_{1\leq i \leq p, 1\leq j \leq
  {\cal N}},\mA)$
is non zero. Letting $\mathscr{O}_\mA\subset \C^{{\cal N}p}$ be the
complement of $V(\Lambda_\mA)$, we deduce that for
$\h\in \mathscr{O}_\mA$, the system ${\frak S}(\h, \mA)$ is
inconsistent, which means that the polynomials $\h^\mA$ satisfy
$\GS(2)$. The transformation
$\varphi: \h\in \C[X_1, \ldots, X_n]^p_\mydegh\mapsto \h^\mA=\h(\mA
\X)\in \C[X_1, \ldots, X_n]^p_\mydegh$
is linear and invertible. The image
$\mathscr{O}_{2}=\varphi(\mathscr{O}_\mA)$ is thus still Zariski open
and satisfies our requirements.

\paragraph*{Genericity of $\GS(3)$.}
We next prove that for a generic choice of $\h$, the polar variety
$W(\pi_1,V(\h))$ is finite. The proof is similar to the one above,
with a few modifications.  This time, we define
$\Delta'({\frak H}, {\frak A})$ to be the vector of $p$-minors of
$\jac({\frak H}^{\frak A}, 1)$, and let
${\frak S}'({\frak H}, {\frak A})\subset \F[X_1, \ldots, X_n]$ be
system of the polynomials
$({\frak H}^{\frak A}, \Delta'({\frak H}, {\frak A})).$ The
polynomials ${\frak S}'(\h, {\frak A})$ and ${\frak S}'(\h, \mA)$ are
defined as above.

Then, we proceed as before, noticing that there exists a non-zero
polynomial $\Lambda'$ in \sloppy
$\C[(\gamma_{i,j})_{1\leq i \leq p, 1\leq j \leq {\cal
    N}},(\alpha_{k,\ell})_{ 1\leq k,\ell\leq n}]$
such that for any $\h,\mA$ that do not cancel $\Lambda'$, the
zero-sets of the systems ${\frak S}'({\frak H}, {\frak A})$ and
${\frak S}'(\h, \mA)$ have the same dimension $r'$, the former being
over an algebraic closure of $\F$.  Fix an $\h$ such that
$\Lambda'(\h,{\frak A})$ is not zero and that satisfies
$\GS(1)$. \cite[Proposition 3.7]{SaSc13} shows that for
$\mA$ in a suitable Zariski open subset of $\C^{n\times n}$,
$W(\pi_1,V(\h^\mA))$ is finite, or equivalently ${\frak S}'(\h, \mA)$
is finite. As for the previous property, this now implies that $r'$ is
either $0$ or $-1$.

In particular, there exists $\mA$ such that ${\frak S}'({\frak H},
\mA)$ has dimension $r'$ as well; thus, this $\mA$ being fixed, we
deduce that there exists an open set $\mathscr{O}'_\mA$ of $\C^{{\cal N}p}$
such that for $\h$ in $\mathscr{O}'_\mA$, $W(\pi_1,V(\h^\mA))$ is finite.
The conclusion follows as in the previous paragraph, by defining
$\mathscr{O}_3=\varphi(\mathscr{O}'_\mA)$.

\paragraph*{Genericity of $\GS(4)$.}
We first prove that for $\h=(h_1, \ldots, h_p)\in \mathscr{O}_1$, the
first claim in $\GS(4)$ holds. Let $m$ be a $(p-1)$-minor of
$\jac(\h, 1)$; without loss of generality, we assume that this minor
is the upper left minor.

Take $\bx$ that cancels all of $\h, {\sf Minors}(\h, m)$, and such
that $m(\bx)\neq 0$; we prove that $\bx$ belongs to $W(\pi_1, V(\h))$.
Indeed, by elementary linear algebra (using Cramer's rule), we deduce that
there exists a non-zero row vector $[\lambda_1, \ldots, \lambda_p]$
such that
\[
h_1(\bx)=\cdots=h_p(\bx)=0, \qquad 
[\lambda_1,\ldots, \lambda_p]\cdot\jac(\h, 1)=[0, \ldots, 0]. 
\]
We deduce that $\jac(\h, 1)$ is rank deficient at $\bx$, and as
pointed out in the statement of $\GS(1)$ given above, this implies
that $\bx$ belongs to $W(\pi_1, V(\h))$.
For the reverse inclusion, take now $\bx\in W(\pi_1, V(\h))\cap {\cal
  O}(m)$. This implies that $\jac(\h, 1)$ is rank deficient at $\bx$,
so that all minors in ${\sf Minors}(\h, m)$ vanish at $\bx$.  Hence,
we proved that in the open set defined by $m\ne 0$, $W(\pi_1, V(\h))$
is the zero-set of $\h, {\sf Minors}(\h, m)$.

Finally, we have to prove that for a generic choice of $\h$, the
Jacobian matrix of the polynomials $\h, {\sf Minors}(\h, m)$ has full
rank $n$ at every point in $W(\pi_1, V(\h))$ where $m$ does not
vanish. The proof is again modeled on the pattern of our proof of
$\GS(2)$.

Consider the polynomials ${\frak S}''({\frak H}, {\frak A})$,
consisting of ${\frak H}^{\frak A}, {\sf Minors}({\frak H}^{\frak A},
m_{\frak A})$, where $m_{\frak A}$ denotes the top-left $(p-1)$-minor
of $\jac({\frak H}^{\frak A},1)$, together with their Jacobian
determinant $C_{\frak A}$ and the polynomials $m_{\frak A} T - 1$,
where $T$ is a new variable.
We first prove that this system has no solution, over an algebraic
closure of~$\F$. 

As we did before, we notice that there exists a non-zero
polynomial~$\Lambda''$ in $\C[(\gamma_{i,j})_{1\leq i \leq p, 1\leq j
    \leq {\cal N}},(\alpha_{k,\ell})_{ 1\leq k,\ell\leq n}]$ such that for
any $\h,\mA$ that do not cancel $\Lambda''$, the zero-sets of the
systems ${\frak S}''({\frak H}, {\frak A})$ and ${\frak S}''(\h, \mA)$
have the same dimension $r''$. Again, we choose $\h$ in
$\mathscr{O}_1$ and such that $\Lambda''(\h,{\frak A})$ is not zero.

For such an $\h$, because $V(\h)$ is smooth, the third and fourth item
of~\cite[Proposition B.1]{SaSc13} prove that for a generic choice of
$\mA$, the Jacobian matrix of $\h^\mA, {\sf Minors}(\h^\mA, m_\mA)$
has full rank $n$ at every point of
$W(\pi_1,V(\h^\mA))\cap \mathcal{O}(m_\mA)$; as a result, for such an
$\mA$, ${\frak G}''(\h, \mA)$ defines the empty set. As before, this
implies that ${\frak G}''({\frak H}, {\frak A})$ defines the empty set
as well. This in turn implies that for a generic choice of $\mA$, the
system ${\frak S}''({\frak H}, \mA)$ defines the empty set. Fixing
such an $\mA$, we deduce that for a generic choice of $\h$,
${\frak S}''(\h, \mA)$ defines the empty set as well; in other words,
$\h^\mA$ satisfies $\GS(4)$. Undoing the change of variables as we did
before proves the last point in $\GS(4)$.


\subsection{A Lagrangian reformulation}

Suppose in all that follows that $\h$ satisfies $\GS$ and let
$V=V(\h)$.  We now show that under assumption $\GS$, we can derive a
Lagrangian formulation for $W(\pi_1,V)$ that still satisfies
regularity properties. In particular, by $\GS(3)$, $W(\pi_1,V)$ is
finite. Also, by $\GS(1)$, $V$ is smooth, $(n-p)$-equidimensional and
$\h$ generates its vanishing ideal. As previously noticed, this
implies that $W(\pi_1, V)$ is defined by 
\[
h_1=\cdots=h_p=0, \qquad \rank(\jac_\bx(\h, 1))\leq p-1. 
\]
For any $\bx$ in this set, by $\GS(2)$, there exists a non-zero vector
$\bell_\bx = [\ell_{\bx,1},\dots,\ell_{\bx,p}]$ in the left nullspace
of $\jac(\h,1)$, and this vector is unique up to a multiplicative
constant.

\begin{proposition}\label{prop:A}
  Suppose that $\bu=(u_1,\dots,u_p) \in \C^n$ is such that
  $ u_1 \ell_{\bx,1} + \cdots + u_p \ell_{\bx,p} \ne 0$ for all $\bx$
  in $W(\pi_1,V)$. Then the sequence of polynomials in variables
  $X_1,\dots,X_n,L_1,\dots,L_p$
  $${\cal W}_\bu = \big ( \h, \quad [L_1 ~\cdots~L_p] \cdot \jac(\h,1), \quad 
  u_1 L_1 + \cdots + u_p L_p - 1 \big )$$
  is such that 
  \[
  Z({\cal W}_\bu) = 
  \{
  (\bx, \bell_\bx)\in \C^{n+p} \mid \bx \in W(\pi_1,V), (\bx, \bell_\bx)\in V({\cal W}_\u)
  \}.
  \]
\end{proposition}

\begin{proof}
  First, take $(\bx,\bell)$ in $V({\cal W}_\bu)$. The fact that $\bx$
  and $\bell$ cancel both $\h$ and
  $[L_1 ~\cdots~L_p] \cdot \jac(\h,1)$ implies that $\bx$ is in
  $W(\pi_1,V)$ and that $\bell = \lambda \bell_\bx$ for some non-zero
  constant $\lambda$. The fact that
  $u_1 \ell_1 +\cdots + u_p \ell_p=1$ implies that
  $(u_1 \ell_{\bx,1} + \cdots + u_p \ell_{\bx,p})\lambda = 1$. Thus,
  we have proved that $V({\cal W}_\bu)$ is contained in the right-hand
  side.

  Conversely, consider a point $(\bx, 1/(u_1 \ell_{\bx,1} + \cdots +
  u_p \ell_{\bx,p}) \ell_\bx)$, for some $\bx$ in $W(\pi_1,V)$; one
  easily sees that it satisfies the defining equations of the zero-set
  ${\cal V}_\bu$, so we have proved that
  $$V({\cal W}_\bu) = \left (\bx, \frac 1{u_1 \ell_{\bx,1} + \cdots +
      u_p \ell_{\bx,p}} \ell_\bx \right )_{\bx \in W(\pi_1,V)} \subset
  \C^{n+p}.$$
  We next prove that all solutions are simple. Take $\bx$ in
  $W(\pi_1,V)$, together with the corresponding $\bell$ such that
  $(\bx,\bell)$ is in ${\cal W}_\bu$.  By $\GS(2)$, there exists a
  $(p-1)$-minor $m_\bx$ of $\jac(\h,1)$ such that $m_\bx(\bx)$ is non-zero;
  let $\iota$ be the index of the missing row. Using Proposition~5.3
  of~\cite{SaSc13}, we deduce the existence of rational
  functions $(\rho_j)_{j =1,\dots,p, j \ne \iota}$ in $\Q[\X]$ such
  that we have equality between ideals
  $$\left \langle
  \h, \quad [L_1 ~\cdots~L_p] \cdot \jac(\h,1) \right \rangle = 
  \left \langle
  \h, \quad L_\iota {\sf Minors}(\h, m_\bx),\quad (L_j - \rho_j L_\iota)_{j =1,\dots,p, j \ne
    \iota}
  \right \rangle$$
  in the localization $\Q[\X,\bL]_{m_\bx}$. Add the equation $  u_1 L_1 + \cdots + u_p L_p - 1$ to
  both sides. On the left, we obtain the equations for
  ${\cal W}_\bu$. On the right, we obtain
  $$   \left \langle
  \h, \quad L_\iota {\sf Minors}(\h, m_\bx),\quad (L_j - \rho_j L_\iota)_{j =1,\dots,p, j \ne
    \iota},\quad  u_1 L_1 + \cdots + u_p L_p - 1
  \right \rangle,$$
  which is equal to
  $$   \left \langle
  \h, \quad L_\iota {\sf Minors}(\h, m_\bx),\quad (L_j - \rho_j L_\iota)_{j =1,\dots,p, j \ne
    \iota},\quad  (u_1 \rho_1 + \cdots + u_p \rho_p)L_\iota - 1
  \right \rangle,$$
  provided we write $\rho_\iota = 1$; this is in turn the same ideal
  as 
  $$ \left \langle \h, \quad {\sf Minors}(\h, m_\bx),\quad (L_j -
  \rho_j L_\iota)_{j =1,\dots,p, j \ne \iota},\quad (u_1 \rho_1 +
  \cdots + u_p \rho_p)L_\iota - 1 \right \rangle.$$ Since
  $(\bx,\bell)$ is in ${\cal W}_\bu$, and $m_\bx(\bx)$ is non-zero,
  $(\bx,\bell)$ must cancel all equations above. In particular,
  $(u_1 \rho_1 + \cdots + u_p \rho_p)(\bx)$ is non-zero.

  Now,  $\GS(4)$ states that the Jacobian matrix of $ ( \h, {\sf
    Minors}(\h, m_\bx))$ has full rank at $\bx$. Writing down that
  Jacobian of the system above in $\Q[\X,\bL]_{m_\bx}$, and using the
  fact that $(u_1 \rho_1 + \cdots + u_p \rho_p)(\bx)$ does not vanish,
  one sees that this larger Jacobian matrix has full rank $n+p$ at
  $(\bx,\bell)$. The equality between ideals seen above implies that
  it is also the case for the polynomials defining ${\cal W}_\bu$.
\end{proof}

The following lemma shows that one can find a suitable
$\bu$ with small bit-size. The proof is a direct application
of Lemma~\ref{lemma:avoid}.

\begin{proposition}\label{prop:U}
  Let $\delta$ be an upper bound on the cardinality of $W(\pi_1,V)$
  and consider the set of linear forms 
  $$u^{(i)} = L_1 + i L_2 + \cdots + i^{p-1} L_p,$$
  for $i$ in $\{1,\dots,8(p-1)\delta\}$. Then at least $7/8$ of these
  linear forms satisfy the assumptions of Proposition~\ref{prop:A}.
\end{proposition}



\subsection{Explicit bound for Lagrange systems: proof of Theorem~\ref{thm:main}}

We continue with the notation introduced at the begining of this
section and let $\myheight$ be an upper bound on the height of all
$h_i$, $i=1,\dots,p$.  Assume that $\h$ satisfies the genericity
assumptions $\GS$ defined previously. As in the
previous subsection, let ${\cal W}_\bu$ be the system
$$
\big ( \h, \quad [L_1 ~\cdots~L_p] \cdot \jac(\h,1), \quad 
  u_1 L_1 + \cdots + u_p L_p - 1 \big )
$$ with $\u$ chosen as in Proposition \ref{prop:U}; we write $\g=(g_1,
  \ldots, g_{n-1})$ for the polynomials $[L_1 ~\cdots~L_p] \cdot
  \jac(\h,1)$ and $\ell= u_1 L_1 + \cdots + u_p L_p - 1$.

The proof of Theorem~\ref{thm:main} simply consists in applying
Theorem~\ref{thm:homoZ} to ${\cal W}_\bu$. Let us review
the quantities that appear in that proposition, and adapt 
them to our present context.

\smallskip

\begin{itemize}
\item We have here $m=2$ and $\bn=(n,p)$.

\smallskip

\item The multi-degrees of the input polynomials in  ${\cal W}_\bu$ are bounded by
  the multi-degree vector
  $\d=(\dunder_1, \ldots, \dunder_1, \dunder_2, \ldots, \dunder_2, \dunder_3)$, with $\dunder_1=(\mydegh,0)$
  appearing $p$ times, $\dunder_2=(\mydegh-1,1)$ appearing $n-1$ times and
  $\dunder_3=(0,1)$ appearing once. Expanding the product
  $$\bchi(\d)=(\mydegh \vartheta_1)^p ((\mydegh-1)\vartheta_1 + \vartheta_2)^{n-1}\vartheta_2 \bmod \langle \vartheta_1^{n+1},\vartheta_2^{p+1}\rangle,$$
  we deduce that $\scrC_\bn(\d)={{n-1}\choose{p-1}}\mydegh^p(\mydegh-1)^{n-p}$.
  Proposition~\ref{prop:degH} then implies that $Z({\cal W}_\bu)$ is a
  finite set of cardinality bounded by this quantity.
  In the particular case $\mydegh=2$, the expression above becomes
 $\scrC_\bn(\d)={{n-1}\choose{p-1}}2^p$.

\smallskip

\item The polynomials $\h$, $\g$ and $\ell$ have heights bounded by
  respectively $\myheight$, $\myheight+\log(n)+\log(\mydegh)$ and
  $p\log(8p \scrC_\bn(\d))$. Using the notation introduced in
 Section~\ref{ssec:heightbds}, we now define
  \begin{align*}
  \eta_1&=\myheight+\mydegh\log(n+1),\\
  \eta_2&=\myheight+\log(n)+\log(\mydegh)+(\mydegh-1)\log(n+1)+\log(p+1),\\ 
  \eta_3&=p\log(8p\scrC_\bn(\d))+\log(p+1).
  \end{align*}
  We can then let
  $\bbeta=(\mu_1, \ldots, \mu_1, \mu_2, \ldots, \mu_2, \mu_3)$,
  with $\mu_1$ appearing $p$ times and $\mu_2$ appearing $n-1$ times.  The
  corresponding arithmetic Chow ring is
  $\R[\myxi,\vartheta_1,\vartheta_2]/\langle
  \myxi^2,\vartheta_1^{n+1},\vartheta_2^{p+1}\rangle$, and we have
$$
\bchi'( \bbeta,\d ) = (\mu_1\myxi+\mydegh\vartheta_1)^{p}
(\mu_2\myxi+(\mydegh-1)\vartheta_1+\vartheta_2)^{n-1}(\mu_3\myxi+\vartheta_2)\bmod
\langle \myxi^2, \vartheta_1^{n+1},\vartheta_2^{p+1} \rangle.
$$
We deduce that
\begin{eqnarray*}
\scrH_\bn(\bbeta,\d) & = & 
     \mu_1 \mydegh^{p-1}(\mydegh-1)^{n-p}
                        \left ({{n-1}\choose{p-1}}+(\mydegh-1){{n-1}\choose{p-2}}\right ) + \\
& &     \mu_2 \mydegh^{p} (\mydegh-1)^{n-p-1}\left ({{n-2}\choose{p-1}}+(\mydegh-1){{n-2}\choose{p-2}}\right )  + \\
 & & \mu_3 \mydegh^{p}(\mydegh-1)^{n-p-1}\left ({{n-1}\choose{p}}+(\mydegh-1){{n-1}\choose{p-1}}\right ) + \\
 & & \mydegh^p(\mydegh-1)^{n-p}{{n-1}\choose{p-1}}.
\end{eqnarray*}
Letting $B_1={{n-1}\choose{p-1}}+{{n-1}\choose{p-2}}={{n}\choose{p-1}}$,
$B_2= {{n-2}\choose{p-1}}+{{n-2}\choose{p-2}}={{n-1}\choose{p-1}}$ and
$B_3={{n-1}\choose{p}}+{{n-1}\choose{p-1}}={{n}\choose{p}}$, we deduce
that
\begin{eqnarray*}
\scrH_\bn(\bbeta,\d) & \leq &  
     \mydegh^{p}(\mydegh-1)^{n-p}\left (
                        \mu_1 B_1+(\mu_2+1)B_2+\mu_3B_3  \right ).
\end{eqnarray*}
Observing that $B_1+B_2+B_3\leq (n+2)B_3$, we obtain the upper bound
\begin{eqnarray*}
\scrH_\bn(\bbeta,\d) & \leq &  \mydegh^{p}(\mydegh-1)^{n-p} \max(\mu_1,\mu_2+1,\mu_3) (n+2) B_3.
\end{eqnarray*}
 This implies that
\[
\scrH_\bn(\bbeta,\d)\in O\tilde{~}\left ( n{{n}\choose{p}}(\myheight+\mydegh)\mydegh^p(\mydegh-1)^{n-p}\right ).
\] 
In the particular case $\mydegh=2$, we obtain
\[
\scrH_\bn(\bbeta,\d)\in O\tilde{~}\left ( n{{n}\choose{p}} \myheight 2^p\right ).
\] 
\item For a general value of $\mydegh$, we will assume that $\h$ is given by
  a straight-line program of length $E$ with constants of height
  bounded by $\myheight'$.  Using Baur-Strassen's algorithm
  \cite{BaurStrassen}, one can deduce a straight-line program with
  constants of bit size in $O(\myheight')$ evaluating $\h$ and $\jac(\h)$
  in time $O(pE)$.  Hence, one can deduce a straight-line program with
  constants of bit size in $O(\myheight')$ evaluating $\h$ and $\g$ in time
  $O(pE+pn)$.
  
  Altogether, the system ${\cal W}_\bu$ can be evaluated by
  straight-line program $\Gamma$ of length $L\in O(pE+pn)$
  with constants of height at most $b=\max(\myheight',
  p\log(8p\scrC_\bn(\d)))$.

  When $\mydegh=2$, we use the obvious construction to construct the
  straight-line program for $\h$ (simply expanding all polynomials on
  the monomial basis), with in this case $E\in O(pn^2)$ and
  $\myheight'=\myheight$.
\end{itemize}

\smallskip

Proposition~\ref{prop:U} ensures that $\u$ is well-chosen with
probability at least $7/8$. Using the fact that the total number of
variables $N$ is at most $2n$, Theorem~\ref{thm:homoZ} shows that
on input $\Gamma$, $\d$ and $\bbeta$, Algorithm {\sf
  NonSingularSolutionsOverZ} runs within
$$ O\tilde{~}\left (p (E+n) \myheight' + \scrC_\bn(\d)\scrH_{\bn}(\bbeta,\d) \left (pE+
    n\mydegh +n^2 \right) n^2\right)$$
boolean operations (the expression given in that proposition also
involves a term of the form
$\log(\max(\mu_1,\mu_2,\mu_3))$, but it is
polylogarithmic in terms of $\scrH_{\bn}(\bbeta,\d)$). It returns the
correct output with probability at least $21/32$, so the overall
probability of success is at least $147/256$, as claimed.  Using the
equalities and inequalities
$$\scrC_\bn(\d) = {{n-1} \choose {p-1}}\mydegh^p (\mydegh-1)^{n-p},\quad
\scrH_\bn(\bbeta,\d)\in O\tilde{~}\left ( n{n \choose
    p}(\myheight+\mydegh)\mydegh^p(\mydegh-1)^{n-p} \right ),$$
the bound on the running time becomes
$$
O\tilde{~}\left ( p (E+n) \myheight' + n^{3}{{n-1} \choose {p-1}}{n
    \choose p}
  (\myheight+\mydegh)\mydegh^{2p}(\mydegh-1)^{2(n-p)}(pE+n\mydegh+n^2)
\right ).$$
In the special case $\mydegh=2$, with $E\in O(pn^2)$ and
$\myheight'=\myheight$, this is
$$ O\tilde{~}\left (n^{5}{{n-1} \choose {p-1}}{n \choose p} 2^{2p} \myheight \right ).$$
The height bound on the coefficients in the output follows immediately
from Theorem~\ref{thm:homoZ} and the bounds on $\scrC_\bn(\d)$ and 
$\scrH_\bn(\bbeta,\d)$.


\section{Proof of Proposition~\ref{prop:hgt}}\label{proof:prop:hgt}

We conclude with the proof of Proposition~\ref{prop:hgt}, which reads 
as follows: {\em
  Let $\f=(f_1,\dots,f_N)$ be polynomials in $\Z[\X_1,\dots,\X_m]$,
  with $\mdeg(\f) \le \d=(\dunder_1,\dots,\dunder_N)$ and
  $\dunder_i=(d_{i,1},\dots,d_{i,m})$ for all $i$, and $\hgt(\f) \le
  \bs=(s_1,\dots,s_N)$; let also $\lambda$ be a separating linear form
  for $Z(\f)$ with integer coefficients of height at most $b$. Then
  all polynomials in the zero-dimensional parametrization of $Z(\f)$
  associated to $\lambda$ have height at most $\scrH_\bn(\bbeta,\d) +
  (b + 4\log(N + 2))\scrC_\bn(\d),$ with}
$$\bbeta=\left (s_i + \sum_{j=1}^m \log(n_j+1) d_{i,j}\right)_{1 \le i \le N}.$$

\paragraph*{The Chow forms.} As a preliminary, we recall the 
definition of the Chow form of an algebraic set. Let $V \subset \Qbar{}^N$
be a zero-dimensional algebraic set. We call {\em Chow form} of $V$
any polynomial of the form
$$C_{V,a} = a \prod_{\bx=(x_1,\dots,x_N) \in V}(T_0 - x_1 T_1 - \cdots
- x_N T_N),$$ for some nonzero $a$ in $\Qbar$. If $V$ is defined over
$\Q$, then for $a$ in $\Q$, $C_{V,a}$ is in
$\Q[T_0,\dots,T_N]$. Clearing denominators and removing contents, we
see that only two of them are primitive polynomials in
$\Z[T_0,\dots,T_N]$ (they differ by a sign): we call them the {\em
  primitive} Chow forms of $V$.

\paragraph*{The arithmetic Chow ring.}
The proof of Proposition \ref{prop:hgt} will rely on objects
introduced by, and results due to, D'Andrea, Krick and
Sombra~\cite{DaKrSo13}. We give here a quick overview of the main
features of their construction.

Introducing new variables $X_{1,0},\dots,X_{m,0}$ as homogenization
variables, we will use $\X'=(\X'_1,\dots,\X'_m)$, with
$\X'_j=(X_{j,0},\dots,X_{j,n_j})$ for all $j$, to describe
multi-homogeneous polynomials.  To any $r$-equidimensional algebraic
set $V \subset \P^\bn$ defined over $\Q$, we associate its class
$[V]_\Z \in A^*(\P^\bn,\Z)$, which takes the form of an homogeneous 
expression of degree $N-r$:
\begin{gather*}
[V]_\Z = \sum_{\bc\, \in \N^m,\ |\bc|=r+1,\, \bc \le \bn}
\widehat{h_\bc}(V)\, \zeta\, \vartheta_1^{n_1-c_1} \cdots
\vartheta_m^{n_m-c_m} + \\\sum_{\bc\, \in \N^m,\ |\bc|=r,\, \bc \le \bn}
\deg_\bc(V) \vartheta_1^{n_1-c_1} \cdots \vartheta_m^{n_m-c_m},  
\end{gather*}
where $\widehat{h_{\bc}}(V)$ and $\deg_{\bc}(V)$ are families of
non-negative real numbers.  For $\bc=(c_1,\dots,c_m)$, the degree
$\deg_{\bc}(V)$ is defined as the generic number of intersection points
between $V$ and $c_1$ linear forms in $\X'_1$, \dots, $c_m$ linear
forms in $\X'_m$. 
The height component $\widehat{h_\bc}(V)$ is harder to define,
and we refer to~\cite{DaKrSo13} for a precise statement (the properties 
given below will be sufficient for our purposes).  
When $V$ has dimension zero, using a slight re-indexing of
the height components, we can write
$$[V]_\Z = \sum_{1 \le i \le m}  \widehat{h_i}(V)\, \zeta\, \vartheta_1^{n_1} \cdots \vartheta_i^{n_i-1} \cdots \vartheta_m^{n_m}
+ 
\deg(V) \vartheta_1^{n_1} \cdots \vartheta_m^{n_m},
$$ where $\widehat{h_i}(V)$ is defined as $\widehat{h_{\bc_i}}(V)$,
with $\bc_i$ the $i$th unit vector, and where $\deg(V)$ is
simply its cardinality.

We now list a few properties which will be central for our purposes.
\begin{itemize}
\item [$\sf A_1.$] For any $V$ as above, $\widehat{h_\bc}(V) \ge 0$ holds
  for all $\bc$~\cite[Proposition~2.51.2]{DaKrSo13}. In other words,
  we have $[V]_\Z \ge 0$, where here, and in all that follows,
  inequalities between elements of arithmetic Chow rings are to be
  understood coefficientwise.
\item [$\sf A_2.$] If $V$ and $V'$ are both $r$-equidimensional and without
  irreducible components in common, $[V \cup V']_\Z=[V]_\Z + [V']_\Z$
  (this is clear for the degree and follows
  from~\cite[Definition~2.40]{DaKrSo13} for the height). We could remove
  the assumption above, but this would require us to talk about cycles,
  for which we will have no use below.
\item [$\sf A_3.$] If $V$ is a hypersurface given as $V=V(f)$, with $f \in
  \Z[\X'_1,\dots,\X'_m]$ multi-homogeneous, squarefree and primitive,
  we have from~\cite[Proposition~2.53]{DaKrSo13}
  $$[V]_\Z = m(f) \zeta + \deg_{\X'_1}(f) \vartheta_1 + \cdots + \deg_{\X'_m}(f) \vartheta_m,$$ 
  where $m(f) = \int_{S_1^{N+m}} \log(|f|) d\mu^{N+m}$ is the {\em Mahler measure}
    of $f$ with respect to the Haar measure $\mu$ of mass 1 on the complex 
    unit circle $S_1$.
  \item [$\sf A_4.$] If $V$ is an $r$-equidimensional algebraic subset  of $\P^\bn$
    defined over $\Q$   and $f$ is multi-homogeneous in $\Z[\X'_1,\dots,\X'_m]$, we have
  from~\cite[Corollary~2.61]{DaKrSo13}
$$[W]_\Z \le [V]_\Z \cdot [f]_{\rm sup},$$
  where $W$ is the $(r-1)$-dimensional part of $V \cap V(f)$,
$|f|_{\rm sup} = {\rm sup}_{\bx \in S_1^{N+m}} |f(\bx)|$
 and
$$[f]_{\rm sup} = \log(|f|_{\rm sup}) \zeta + \deg_{\X'_1}(f) \vartheta_1 +\cdots+
\deg_{\X'_m}(f)\vartheta_m.$$
\end{itemize}

\paragraph*{Using the B\'ezout inequality.} Let $\f^h=(f_1^h,\dots,f_N^h)$
be the polynomials in $\Z[\X'_1,\dots,\X'_m]$ obtained by
multi-homogenizing the input $f_1,\dots,f_N$ with respect to all
groups of variables $\X_1,\dots,\X_m$, let $S\subset \P^\bn$ be the
zero-dimensional component of $V(\f^h)$, and let
$\d=(\dunder_1,\dots,\dunder_N)$ and $\s=(\myheight_1,\dots,\myheight_N)$ be
upper bounds on respectively $\mdeg(\f)$ and $\hgt(\f)$; as in the
proposition, we define
$$\bbeta=(\eta_1,\dots,\eta_N)=\left (s_i + \sum_{j=1}^m \log(n_j+1) d_{i,j}\right)_{1 \le i \le N}.$$

By~\cite[Proposition~2.51.3]{DaKrSo13}, $[\P^\bn]_\Z=1$. Applying
${\sf A}_4$ repeatedly, we obtain that
$$[S]_\Z \le [f_1^h]_{\rm sup} \cdots [f_N^h]_{\rm sup}.$$
By~\cite[Lemma~2.32]{DaKrSo13}, for all $i$, we have the inequality
$$[f_i]_{\rm sup} \le \eta_i \zeta + d_{i,1} \vartheta_1 + \cdots +
d_{i,m} \vartheta_m,$$ or equivalently $[f_i]_{\rm sup} \le
\chi'(\eta_i,\dunder_i)$. This implies that 
\begin{equation}\label{eq:S}
[S]_\Z \le  \chi'(\eta_1,\dunder_1) \cdots \chi'(\eta_N,\dunder_N) = \bchi'(\bbeta,\d).  
\end{equation}

\paragraph*{From multi-projective to affine.}
Let now $S'\subset \P^{\bn}$ be the subset of $S$ consisting of all those points
$\bx'=(\bx'_1,\dots,\bx'_m)$ in $S$, with $\bx'_i$ in
$\P^{n_i}(\Qbar)$ for all $i$, such that

\smallskip

\begin{itemize}
\item $\bx'_i$ does not belong to the hyperplane at
  infinity in $\P^{n_i}(\Qbar)$;
\smallskip
\item the multi-homogeneous polynomial $J^h$ obtained by
  multi-homogenizing the Jacobian determinant $D=\det(jac(\f))$ with
  respect to all groups of variables $\X_1,\dots,\X_m$ does not vanish
  at $\bx'$.
\end{itemize}
\smallskip \noindent
Because we obtain $S'$ by removing algebraic subsets from $S$, and
these subsets are defined over $\Q$, $S'$ itself is defined over $\Q$.
Using ${\sf A}_1$ and ${\sf A}_2$, we deduce from~\eqref{eq:S} that
we have
\begin{equation}\label{eq:Sp}
[S']_\Z \le   \bchi'(\bbeta,\d).  
\end{equation}

Our goal is now to compute the Chow form of the related algebraic $Z(\f)$ in
$\Qbar{}^N$. For $(\bx'_1,\dots,\bx'_m)$ in $S'$, our definition shows
that each block-coordinate $\bx'_i$ can be written as
$\bx'_i=(1,x_{i,1},\dots,x_{i,n_i})$. We use this notation in the
lemma below --- whose proof is a direct consequence of our
construction.

\begin{lemma}
  The following equality holds
  $$Z(\f)=\{(x_{1,1},\dots,x_{1,n_1},\dots,x_{m,1},\dots,x_{m,n_m})
  \ \mid \ \bx \in S'\} \subset \Qbar{}^N.$$ 
\end{lemma}
Letting $T_0,\dots,T_{N}$ be new variables, the Chow
forms of $Z(\f)$  are thus of the form 
\begin{equation}\label{eq:ChowZf}
C_{Z(\f),c} = c \, \prod_{\bx \in S'} (T_0 - x_{1,1} T_1 - \dots -
x_{m,n_m-1} T_{N -1} - x_{m,n_m-1} T_{N} ),  
\end{equation}
 for
some constant $c$. 

Let us next describe a classical geometric way to construct these Chow
forms starting from $S'$. We start by considering the product
$\mathscr{T} = S' \times \P^{N}(\Qbar)$, which is an algebraic subset
of $\P^{\bn} \times \P^{N}(\Qbar)$; we use $T_0,\dots,T_{N}$ as our
coordinates in $\P^{N}(\Qbar)$.  Next, define $\mathscr{T}'$ as the
intersection of $\mathscr{T}$ and $Z(K^h)$, where $K$ is given by
$$K=T_0 -( X_{1,1} T_1 + X_{1,2} T_2 + \cdots + X_{m,n_m-1} T_{N-1} + X_{m,n_m} T_{N}) $$ and $K^h$ is
obtained by multi-homogenizing $K$ with respect to the groups of variables
$\X_1,\dots,\X_m$, using respectively $X_{1,0},\dots,X_{m,0}$ ($K$ is
already homogeneous with respect to $T_0,\dots,T_{N}$). 

\begin{lemma}
  The intersection $\mathscr{T}'=\mathscr{T} \cap V(K^h)$ is proper.
\end{lemma}
\begin{proof}
  Since $S'$ is finite, it is sufficient to consider the case where
  $S'$ is a single point of the form $(\bx'_1,\dots,\bx'_m)$. In that
  case, the set $\mathscr{T}'$ is isomorphic to the zero-set of the
  linear form $K^h(\bx'_1,\dots,\bx'_m,T_0,\dots,T_{N})$ in
  $\P^{N}(\Qbar)$. Our construction of $S'$ implies that the
  coefficient of $T_0$ in this linear form is non-zero, so we are
  done.
 \end{proof}

Finally, call $\pi$ the projection on the last factor $\P^{N}(\Qbar)$,
and let us define $\mathscr{Y}$ as the image of $\mathscr{T}'$ by this projection.
\begin{lemma}
  The image of each $\Qbar$-irreducible component of $\mathscr{T}'$ by $\pi$ is
  a hypersurface and each squarefree polynomial in $\Q[T_0,\dots,T_N]$
  defining $\mathscr{Y}$ is a Chow form of $Z(\f)$.
\end{lemma}
\begin{proof}
  Continuing the proof of the previous lemma, we see that the
  $\Qbar$-irreducible components of $\mathscr{T}'$ are finite unions of sets of
  the form $(\bx'_1,\dots,\bx'_m)\times H$, where, writing
  $\bx'_i=(1,x_{i,1},\dots,x_{i,n_i})$, $H$ is the hyperplane
  of $\P^{N}(\Qbar)$ defined by
$$K=T_0 -( x_{1,1} T_1 + x_{1,2} T_2 + \cdots + x_{m,n_m-1} T_{N-1} + x_{m,n_m} T_{N}).$$
  The conclusion follows from~\eqref{eq:ChowZf}.
\end{proof}

\paragraph*{Explicit bounds.} We can now give quantitative
estimates for the classes of the objects introduced so far.
By~\cite[Proposition~2.66]{DaKrSo13}, we have the equality $[\mathscr{T}]_\Z =
\iota([S']_\Z)$, where $[\mathscr{T}]_\Z$ lies in $A^*(\P^{\bn} \times
\P^{N}(\Qbar),\Z)$ and $\iota$ is
the canonical injection
\begin{multline*}
A^*(\P^\bn,\Z)=\R[\zeta, \vartheta_1, \dots,\vartheta_m]/\langle \zeta^2,
\vartheta_1^{n_1+1},\dots,\vartheta_m^{n_m+1} \rangle\\
\to
A^*(\P^{\bn} \times
\P^{N}(\Qbar),\Z)=\R[\zeta, \vartheta_1, \dots,\vartheta_m,\mu]/\langle \zeta^2,
\vartheta_1^{n_1+1},\dots,\vartheta_m^{n_m+1},\mu^{N+1} \rangle.
\end{multline*}
Since $S'$ has dimension zero, its class in $A^*(\P^\bn,\Z)$ 
has the form
\begin{equation}\label{eq:defSp}
[S']_\Z = \sum_{1 \le i \le m}  \widehat{h_i}(S')\, \zeta\, \vartheta_1^{n_1} \cdots \vartheta_i^{n_i-1} \cdots \vartheta_m^{n_m}
+ 
\deg(S') \vartheta_1^{n_1} \cdots \vartheta_m^{n_m}.
\end{equation}
 We deduce that
$[\mathscr{T}]_\Z$ has the same form, but in $A^*(\P^{\bn} \times
\P^{N}(\Qbar),\Z)$.
Remark next that the element $[K^h]_{\rm sup} \in A^*(\P^{\bn} \times
\P^{N}(\Qbar),\Z)$
satisfies
$$[K^h]_{\rm sup} = \log(N+1) \zeta + \vartheta_1 +\cdots+ \vartheta_m + \mu.$$
Hence, because the intersection defining $\mathscr{T}'$ is proper, we deduce from the B\'ezout inequality $\sf
A_4$ that
$$[\mathscr{T}']_\Z  \le [\mathscr{T}]_\Z \cdot (\log(N+1) \zeta + \vartheta_1 +\cdots+ \vartheta_m + \mu).$$
Using the formula for $[\mathscr{T}]_\Z$ given above, we obtain
\begin{align*}
[\mathscr{T}']_\Z &\le 
 \sum_{1 \le i \le m}  \widehat{h_i}(S')\, \zeta\, \vartheta_1^{n_1} \cdots \vartheta_m^{n_m}
 + 
 \sum_{1 \le i \le m}  \widehat{h_i}(S')\, \zeta\, \vartheta_1^{n_1} \cdots \vartheta_i^{n_i-1} \cdots \vartheta_m^{n_m} \mu
\\[1mm]
& ~+ \log(N+1) \deg(S') \zeta\,\vartheta_1^{n_1} \cdots \vartheta_m^{n_m}
+
 \deg(S') \vartheta_1^{n_1} \cdots \vartheta_m^{n_m}\mu.
\end{align*}
Finally, we consider the projection on $\P^{N}(\Qbar)$.
The arithmetic Chow ring of
this projective space is $\R[\zeta,\mu]/\langle \zeta^2,
\mu^{N+1}\rangle$, and~\cite[Proposition~2.64]{DaKrSo13}
shows that 
$$  \vartheta_1^{n_1} \cdots \vartheta_m^{n_m} [\mathscr{Y}]_\Z \le [\mathscr{T}']_\Z.$$
Considering the possible monomial support of $[\mathscr{Y}]_\Z$, we deduce 
that we have the inequality
$$[\mathscr{Y}]_\Z \le  \sum_{1 \le i \le m} \widehat{h_i}(S')\,
\zeta + \log(N+1) \deg(S') \zeta + \deg(S')\mu. $$ Hence, if $C$ is a primitive polynomial in
$\Z[T_0,\dots,T_{N}]$ defining $\mathscr{Y}$, we deduce from
${\sf A_3}$ that $$m(C)\le \sum_{1 \le i \le m} \widehat{h_i}(S') + \log(N + 1)
\deg(S').$$ 
This leads us to the following lemma.
\begin{lemma}\label{prop:chow}
  Any primitive Chow form $C$ of $V(\f)$ satisfies $$m(C)\le \scrH_\bn(\bbeta,\d) + \log(N + 1)
\scrC_\bn(\d).$$
\end{lemma}
\begin{proof}
  In view of the previous discussion, it is enough to prove that the
  inequality $$\sum_{1 \le i \le m} \widehat{h_i}(S') + \log(N + 1)
  \deg(S') \le \scrH_\bn(\bbeta,\d) + \log(N + 1) \scrC_\bn(\d)$$
  holds. We saw in~\eqref{eq:Sp} the inequality $[S']_\Z \le
  \bchi(\bbeta,\d)$, which is to be understood coefficient-wise. Take
  the sum of coefficients on both sides.  From~\eqref{eq:defSp},
  we deduce that the left-hand side adds up to $\sum_{1 \le i \le m}
  \widehat{h_i}(S') +\deg(S')$, which is an upper bound on $\sum_{1
    \le i \le m} \widehat{h_i}(S')$, whereas the right-hand side gives
  $\scrH_\bn(\bbeta,\d)$. To conclude, we add $\log(N + 1) \deg(S')$
  on both sides, and we use the fact that $\deg(S')=\deg(Z(\f)) \le
  \scrC_\bn(\d)$, as pointed out after Proposition~\ref{prop:degH}.
\end{proof}

\paragraph*{Conclusion.} Finally, we can conclude the proof
of Proposition~\ref{prop:hgt}. Lemma~\ref{prop:chow} shows that for
any primitive Chow $C$ form of $Z(\f)$, we have $m(C) \le
\scrH_\bn(\bbeta,\d) + \log(N + 1) \scrC_\bn(\d)$; using the inequality
$|m(C) - {\hgt}(C)| \le \log(N +2)\deg(C)$
(see~\cite[Lemma~2.30]{DaKrSo13}), we deduce that such a Chow form has
height at most $\scrH_\bn(\bbeta,\d) + 2\log(N + 2) \scrC_\bn(\d)$.
Using Lemma~\ref{lemma:hgtQ} below (which is itself a standard result), we deduce that all polynomials
appearing in the zero-dimensional parametrization of $Z(\f)$
associated to a linear form $\lambda$ of height $b$ have height at
most
$$\scrH_\bn(\bbeta,\d) + (b + 4\log(N + 2))\scrC_\bn(\d),$$
which proves the proposition.

\begin{lemma}\label{lemma:hgtQ}
  Suppose that $V \subset \Qbar{}^N$ is a zero-dimensional algebraic set
  defined over $\Q$ and that $\lambda$ is a separating linear form for
  $V$ with integer coefficients of height at most $b$. Suppose as
  well that the primitive Chow forms of $V$ have height at most
  $h$. Then, all polynomials that appear in the zero-dimensional
  parametrization $\scrQ=((q,v_1,\dots,v_N),\lambda)$ of $V$ have
  height at most $h + \log(\deg(V)) + \deg(V)(b+\log(N+1))$.
\end{lemma}
\begin{proof}
  Let $C$ be a primitive Chow form of $V$, written $C = a C_0$, with
  $C_0$ monic in $T_0$. It is well-known (see for instance~\cite{ABRW})
  that we obtain $q$ and $v_1,\dots,v_n$ as 
  $$q = \frac 1a C(T,\lambda_1,\dots,\lambda_n),\quad v_i = -\frac 1a
  \frac{\partial C}{\partial T_i}(T,\lambda_1,\dots,\lambda_n).$$
  Since $C$ has degree $\deg(V)$ and height $h$, its partial
  derivatives have height at most $h+\log(\deg(V))$. The conclusion
  then follows from (for instance) Lemma~1.2.1.c in~\cite{KrPaSo01}.
\end{proof}


\end{document}